\documentclass[a4paper,envcountsame]{amsart}

\usepackage{amsaddr}
\usepackage{comment}

\usepackage[T1]{fontenc}
\usepackage[utf8]{inputenc}

\usepackage{amsmath,amssymb,amsthm}

\newtheorem{theorem}{Theorem}[section]

\newtheorem{proposition}[theorem]{Proposition}
\newtheorem{lemma}[theorem]{Lemma}

\theoremstyle{remark}
\newtheorem{example}{Example}[section]

\usepackage{hyperref}
\usepackage{enum}
\usepackage{subcaption}

\usepackage{comment}

\usepackage{xspace} 
\usepackage{mathtools}

\newcommand{\tuple}[1]{( #1 )}
\def\abs#1{\ensuremath{\lvert #1\rvert}} 

\let\epsilon\varepsilon

\usepackage{gastex}
\usepackage{cutwin}

\newenvironment{longversion}{}{}
\newenvironment{shortversion}{}{}

\newcommand{\calA}{\mathcal{A}}
\newcommand{\calF}{\mathcal{F}}
\newcommand{\calR}{\mathcal{R}}
\newcommand{\calM}{\mathcal{M}}

\newcommand{\calU}{\mathcal{U}}

\newcommand{\prl}{\mathbin{\Vert}}

\renewcommand{\a}{a}
\renewcommand{\b}{b}

\newcommand{\bi}[2]{{}^{#1}_{#2}}
\newcommand{\tbi}[2]{\bi{#1}{#2}}   
\newcommand{\fbi}[2]{\mbox{{\large $\textstyle \bi{#1}{#2}$}}}   
\newcommand{\cfbi}[2]{\fbi{#1}{#2}}  

\newcommand{\reject}{q_{\mathrm{rej}}}

\newcommand{\Obs}{\Sigma}
\newcommand{\sym}{\mathsf{tr}}
\newcommand{\Reflexive}{\mathsf{Ref}}
\newcommand{\Ambiguous}{\mathsf{Amb}}

\renewcommand{\index}{\mathsf{idx}}
\renewcommand{\matrix}{\mathsf{mat}}
\newcommand{\lex}{\mathsf{lex}}
\newcommand{\transform}{\mathsf{next}}
\newcommand{\successor}{\mathsf{succ}}

\def\presuper#1#2%
  {\mathop{}%
   \mathopen{\vphantom{#2}}^{#1}%
   \kern-\scriptspace%
   #2}

\newcommand*{\dfa}{\textsc{dfa}\xspace}
\newcommand*{\dfal}{{Mealy automaton}\xspace}
\newcommand*{\init}{\varepsilon}

\title[Observation and Distinction]{Observation and Distinction \\
  Representing Information in Infinite Games}

\author[D. Berwanger, L. Doyen]{Dietmar Berwanger \and Laurent Doyen}
\address{{\small LSV, CNRS \& ENS} Paris-Saclay, France}
\email{dwb@lsv.fr, doyen@lsv.fr}
\date{\today ; extended version of a contribution to the Proceeding of STACS~2020}

\begin{document}
\excludecomment{shortversion}

\maketitle

\begin{abstract}
We compare two approaches for modelling imperfect information in infinite games
by using finite-state automata. The first, more standard approach views information
as the result of an observation process driven by a sequential Mealy machine.
In contrast, the second approach features indistinguishability relations 
described by synchronous two-tape automata.

The indistinguishability-relation model turns out to be
strictly more expressive than the one based on observations. 
We present a characterisation of the indistinguishability relations that 
admit a representation as a finite-state observation function.
We show that the characterisation is decidable, and give
a procedure to construct a corresponding Mealy machine whenever one exists.

\end{abstract}

\section{Introduction}

Uncertainty is a main concern in strategic interaction.
Decisions of agents are based on their knowledge about the system state,
and that is often limited. 
The challenge grows
in dynamical systems,
where the state changes over time, 
and it becomes severe, when the dynamics 
unravels over infinitely many stages. 
In this context, one fundamental question is how to model knowledge 
and the way it changes as information is acquired along the stages of the 
system run.

Finite-state automata offer a solid framework 
for the analysis of systems with infinite runs. 
They allow to reason about infinite state spaces in terms of finite
ones\,---\,of course, with a certain loss. 
The connection has proved to be extraordinarily successful in the study
of infinite games on finite graphs, in the particular setting of \emph{perfect 
information} assuming that players are informed about
every move in the play history, 
which determines the actual state of the system.    
One key insight is that winning strategies, in this setting, 
can be synthesized effectively \cite{BuechiLan69,Rabin72}:
for every game described by finite automata, 
one can describe the set of winning strategies 
by an automaton (over infinite trees) and, moreover, 
construct an automaton (a finite-state Moore machine) 
that implements a winning strategy.

In this paper, we discuss two approaches for modelling 
\emph{imperfect information}, where,
in contrast to the perfect-information setting, it is no longer assumed
that the decision maker is informed about the moves 
that occurred previously in the play history.

The first, more standard approach corresponds to viewing 
information as a result of an
observation \emph{process} that may be imperfect in the sense that
different moves can yield the same observation in~a stage of the game. 
Here, we propose a second approach, 
which corresponds to representing information 
as a \emph{state} of 
knowledge, by describing which histories 
are indistinguishable to the decision maker. 

Concretely, we assume a setting of 
synchronous games with perfect recall in~a partitional information model.
Plays proceed in infinitely many stages, each of which  
results in one move from a finite range. Histories and plays are thus 
determined as finite or infinite sequences of moves, respectively.
 
To represent information partitions, we consider two models based on
finite-state automata.
In the observation-based model, which corresponds to the standard approach in 
computing science and non-cooperative game theory, 
the automaton is a sequential Mealy machine that inputs moves  
and outputs observations from a finite alphabet. 
The machine thus describes an observation function, which
maps any history of moves 
to a sequence of observations that represents its information set.
In the indistinguishability-based model, we use two-tape automata 
to describe which pairs of histories belong to the same information set.

As an immediate insight, we point out that, in the finite-state setting, 
the standard model based on observation functions 
is less expressive than the one based on indistinguishability 
relations. Intuitively, this is because observation functions can only yield a 
bounded amount of information in each round\,---\,limited by the size of the 
observation alphabet, 
whereas indistinguishability relations can describe situations  where the 
amount of information received per round
grows unboundedly as the play proceeds.

We investigate the question whether an information partition represented
as (an indistinguishability relation given by) a two-tape automaton 
admits a representation 
as (an observation function given by) a Mealy machine.  
We show that this question is decidable, using results from the theory of 
word-automatic structures. We also present a procedure for constructing a
Mealy machine that represents a
given indistinguishability relation as an observation function,
whenever this is possible.

\section{Basic Notions}

\subsection{Finite automata}

To represent components of infinite games as finite objects,
finite-state automata offer a versatile framework (see~\cite{LNCS2500}, for a survey).
Here, we use automata of two different types, which we introduce
following the notation of \cite[Chapter 2]{Mikolajczak91}. 
 
As a common underlying model, a \emph{semi-automaton} is a 
tuple $\calA = (Q, \Gamma, q_\init, \delta)$ 
consisting of a finite set~$Q$ of \emph{states}, 
a finite \emph{input alphabet}~$\Gamma$, 
a designated \emph{initial state} $q_\init \in Q$, 
and a \emph{transition function} $\delta\colon Q \times \Gamma \to Q$.  
We define the size~$\abs{\calA}$ of~$\calA$ to be the number of its transitions, that is $\abs{Q} \cdot \abs{\Gamma}$.
To describe the internal behaviour of the semi-automaton we extend the transition 
function from letters to input words: 
the extended transition function
$\delta\colon Q \times \Gamma^* \to Q$ is defined by setting,  
for every state~$q \in Q$, 
\begin{itemize}
\item $\delta( q, \init ) := q$ for the empty word~$\epsilon$, and
\item  
$\delta( q, \tau c) := \delta( \delta( q, \tau), c )$, for any word 
obtained by the concatenation of a word $\tau \in \Gamma^*$ 
and a letter $c \in \Gamma$. 
\end{itemize}
On the one hand, we use automata as acceptors of finite words.
A \emph{deterministic finite automaton} (for short, $\dfa$) is a
tuple $\calA = (Q, \Gamma, q_\init, \delta, F)$ 
expanding a semi-automaton by a 
designated subset~$F \subseteq Q$ of \emph{accepting states}. 
We say that a finite input word $\tau \in \Gamma^*$ is \emph{accepted} 
by~$\calA$ from a state $q$ if $\delta(q, \tau ) \in F$. 
The set of words in $\Gamma^*$ that are accepted by~$\calA$ from the
initial state $q_\init$ forms its \emph{language}, denoted $L(\calA) \subseteq \Gamma^*$.

Thus, a \dfa recognises a set of words. 
By considering input alphabets over 
pairs of letters from a basis alphabet~$\Gamma$, the model can be 
used to recognise synchronous relations over $\Gamma$, 
that is, relations between words of the same length. 
We refer to a \dfa over an input alphabet $\Gamma \times \Gamma$ 
as a \emph{two-tape} \dfa. The 
relation recognised by such an automaton consists of all pairs 
of words $c_1 c_2 \ldots c_\ell, c'_1 c'_2 \ldots c'_\ell \in \Gamma^*$ 
such that $(c_1, c_1') (c_2, c_2') \ldots (c_\ell, c_\ell') \in L(\calA)$. 
With a slight abuse of notation, 
we also denote this relation by $L(\calA)$. 
We say that a synchronous relation is regular if it is recognised by a \dfa.

On the other hand, we consider automata with output. 
A \emph{Mealy} automaton is
a tuple  $(Q, \Gamma, \Sigma, q_\init, \delta, \lambda)$ 
where $(Q, \Gamma, q_\init, \delta) $ is a semi-automaton, $\Sigma$ 
is a finite \emph{output alphabet}, and 
$\lambda\colon Q \times \Gamma \to \Sigma$ is an output function.
To describe the external behaviour of such an automaton,
we define the extended output function 
$\lambda\colon \Gamma^* \times \Gamma \to \Sigma$
by setting $\lambda( \tau, c) := \lambda( \delta( q_\init, \tau), c )$
for every word $\tau \in \Gamma^*$ and every letter $c \in \Gamma$.
Thus, the external behaviour of a Mealy automaton 
defines a function from the set $\Gamma^+ := \Gamma^* \setminus \{\epsilon \}$ of nonempty words
to $\Sigma$.  
We say that a function on $\Gamma^+$ 
is \emph{regular}, if there exists a Mealy automaton
that defines it.

To build new automata from given ones, we will use two types of product constructions.
The \emph{synchronised product} of two semi-automata 
$\calA^1 = (Q^1, \Gamma, q^1_\init, \delta^1)$
and $\calA^2 = (Q^2, \Gamma, q^2_\init, \delta^2)$, 
over the same alphabet $\Gamma$, 
is the semi-automaton
$\calA^1 \times \calA^2 = (Q^\times,\Gamma, q^\times_\init, \delta^\times)$
with: 
\begin{itemize}
\item $Q^\times = Q^1 \times Q^2$, 
\item $q^\times_\init = \tuple{q^1_\init,q^2_\init}$, and
\item $\delta^\times(\tuple{q^1,q^2},c) 
= \tuple{\delta^{1}(q^1,c),\delta^{2}(q^2,c)}$
for all $q^1 \in Q^1$, $q^2 \in Q^2$, and $c \in \Gamma$.
\end{itemize}

In the second type of product construction, the two automata
run in parallel on separate input tapes, one for each automaton.
There is no synchronisation other than the number of processed input symbols, 
which is always the same in the two automata.
The \emph{parallel product} of two semi-automata 
$\calA^1 = (Q^1, \Gamma^1, q^1_\init, \delta^1)$
and $\calA^2 = (Q^2, \Gamma^2, q^2_\init, \delta^2)$ 
is the semi-automaton 
$\calA^1 \prl \calA^2 = \tuple{Q^\Vert,\Gamma^1 \times \Gamma^2,q^\Vert_I,\delta^\Vert}$
where: 
\begin{itemize}
\item $Q^{\Vert} = Q^1 \times Q^2$, 
\item $q^{\Vert}_\init = \tuple{q^1_\init,q^2_\init}$, and
\item $\delta^{\Vert}(\tuple{q^1,q^2},\tuple{c^1,c^2}) 
= \tuple{\delta^{1}(q^1,c^1),\delta^{2}(q^2,c^2)}$
for all $q^i \in Q^i$ and $c^i \in \Gamma^i$ (with $i=1,2$).
\end{itemize}

\subsection{Repeated games with imperfect information}

In our general setup, we consider games played in an infinite sequence of stages. 
In each stage, every player chooses
an action from a given set of alternatives, independently and simultaneously.
As a consequence, this determines a move that is
recorded in the play history. Then,
the game proceeds to the next stage.
The outcome of the play is thus an infinite sequence of moves.

Decisions of a player are based on the available information, which we
model by a partition of the set of play histories into information sets: 
at the beginning of each stage game, the player is informed of the information set
to which the
actual play history belongs (in the partition associated to the player). 
Accordingly, a strategy for a player is 
a function from information sets to actions.
Every strategy profile (that is, a collection of strategies, one for each player)
determines a play.

Basic questions in this setup concern strategies of an individual player to enforce an outcome in a
designated set of winning plays or to maximise the value of a given payoff function, 
regardless of the strategy of other players.
More advanced issues target joint strategies of coalitions among players towards
coordinating on a common objective, or equilibrium profiles.   
Scenarios where the available actions depend on the history,
or where the play might end after finitely many stages, can be captured by adjusting the
information partition together with the payoff or winning condition.

For our formal treatment of information structures,
we use the model of abstract infinite games as introduced by Thomas in his seminal paper on
strategy synthesis~\cite{Thomas95};
the relevant questions for more elaborate settings, such as infinite games on finite graphs
or concurrent game structures can be reduced easily to this abstraction.  
The underlying model is consistent with the classical definition of extensive games with 
information partitions and perfect recall due to 
von Neumann and Morgenstern \cite{vNM44}, in the formulation of   
Kuhn~\cite{Kuhn53}. 
For a more detailed account on partitional information, we refer 
to Bacharach~\cite{Bacharach85} and Geanakoplos~\cite{Geanakoplos92}.   

Our formalisation captures the 
information structures of
repeated games with imperfect monitoring as studied in 
non-cooperative game theory
(see the survey of Gossner and Tomala~\cite{GossnerTom09}),
and of infinite games with partial observation on finite-state systems
as studied in computing science (see Reif~\cite{Reif84},
Lin and Wonham~\cite{LinWon1988}, van der Meyden and Wilke~\cite{vanderMeydenWil2005},
Chatterjee et al.~\cite{ChatterjeeDHR07},
Berwanger et al.~\cite{BerwangerKP11}).
For background on the modelling of knowledge,
and the notion of synchronous perfect recall we refer to
Chapter 8 in the book of Fagin et al.~\cite{FaginHMV03}.

\subsubsection{Move and information structure}

As a basic object for describing a game,
we fix a finite set $\Gamma$ of \emph{moves}.
A~\emph{play} is an infinite sequence 
of moves $\pi = c_1 c_2 \ldots \in \Gamma^\omega$. 
A \emph{history} (of length $\ell$) is a finite prefix 
$\tau = c_1 c_2 \ldots c_\ell \in \Gamma^*$ of a play; 
the empty history $\init$ has length zero. 
The \emph{move structure} of the game is the set $\Gamma^*$ of histories equipped with
the successor relation, which consists of all pairs $(\tau, \tau c)$
for $\tau \in \Gamma^*$ and $c \in \Gamma$.
For convenience, we denote
the move structure of a game on $\Gamma$ simply by $\Gamma^*$ omitting
the (implicitly defined) successor relation.


The information available to a player 
is modelled abstractly by a  
partition $\calU$ of the set $\Gamma^*$
of histories; 
the parts of $\calU$ are called \emph{information sets} (of the player). 
The intended meaning is that if the actual history belongs to an 
information set~$U$, then the player considers every history in $U$ 
possible.
The particular case where all information sets in the partition 
are singletons characterises the setting of \emph{perfect information}. 

The \emph{information structure} (of the player)
is the quotient $\Gamma^*/_\calU$ of the move structure by the
information partition. That is, the first-order structure on the domain
consisting of the information sets, with a binary relation connecting
two information sets $(U, U')$ 
whenever there exists a history $\tau \in U$ 
with a successor history $\tau c \in U'$.
Throughout this article, we assume the perspective of just one player,
so we simply refer to the information structure of the game.

Our information model is \emph{synchronous}, which means, intuitively, 
that the player always knows how many stages have been played. 
Formally, this amounts to asserting that 
all histories in an information set have the same length; 
in particular the empty history forms a singleton information set.
Further, we assume that the player has \emph{perfect recall}\,---\, 
he never forgets what he knew previously.
Formally, if an information set contains nonempty 
histories $\tau c$ and $\tau' c'$, 
then the predecessor history $\tau$ is in the same information set as 
$\tau'$.
In different terms, an information partition satisfies 
synchronous perfect recall if,
whenever a pair of histories
$c_1 \ldots c_\ell$ and $c'_1 \ldots c'_\ell$ belongs to an information set,
then 
for every stage $t \le \ell$,
the prefix histories $c_1 \ldots c_t$ and $c'_1 \ldots c'_t$ 
belong to the same information set.
As a direct consequence, the information structures that arise from such partitions
are indeed trees.

\begin{lemma}\label{lem:info-tree}
For every information partition~$\calU$ of perfect synchronous recall, 
the information structure $\Gamma^*/_\calU$ is a directed tree.
\end{lemma}

We will use the term \emph{information tree} when referring to the information structure
associated with an information partition with synchronous perfect recall.

In the following, we discuss two alternative representations of information partitions.

\subsubsection{Observation}

The first alternative consists in describing
the information received by the player in each stage.
To do so, we specify a set~$\Sigma$ of \emph{observation symbols} and 
an \emph{observation function} $\beta\colon \Gamma^+ \to \Sigma$.
Intuitively, the player observes at every nonempty history~$\tau$ 
the symbol $\beta(\tau)$; 
under the assumption of perfect recall, 
the information available to the player at history~$\tau = c_1 c_2 \ldots c_\ell$ 
is thus represented by the 
sequence of observations
$\beta( c_1 ) \beta (c_1 c_2) \ldots \beta( c_1 \ldots c_\ell)$, which we call 
\emph{observation history} (at $\tau$); let us denote by 
$\hat{\beta} \colon \Gamma^* \to \Sigma^*$ the function 
that returns, for each play history, the corresponding observation history.

The information partition $\calU_\beta$ represented by 
an observation function~$\beta$ is the collection of sets
$U_\eta := \{ \tau \in \Gamma^* \mid \hat{\beta}( \tau ) = \eta \}$
indexed by observation 
histories $\eta \in \hat{\beta}( \Gamma^* )$.
Clearly, information partitions described in this way 
verify the conditions of synchronous perfect recall:
each information set $U_\eta$ consists of
histories of the same length (as $\eta$), 
and for every pair~$\tau, \tau'$ of histories with different observations 
$\hat{\beta}(\tau) \neq \hat{\beta}(\tau')$, and every pair of moves~$c, c' \in \Gamma$, the observation history of
the successors $\tau c$ and $\tau' c'$ will also differ
$\hat{\beta}( \tau c ) \neq \hat{\beta}( \tau' c')$.

To describe observation functions by a finite-state automaton,  
we fix a \emph{finite} set~$\Sigma$ of observations and 
specify a Mealy automaton 
$\calM = (Q, \Gamma, \Sigma, q_\init, \delta, \lambda)$, 
with moves from  $\Gamma$ as input and observations from~$\Sigma$ as output.
Then, we consider the extended output function of $\calM$
as an observation function
$\beta_\calM \colon\Gamma^+ \to \Sigma$.

\begin{figure}[!tbp]
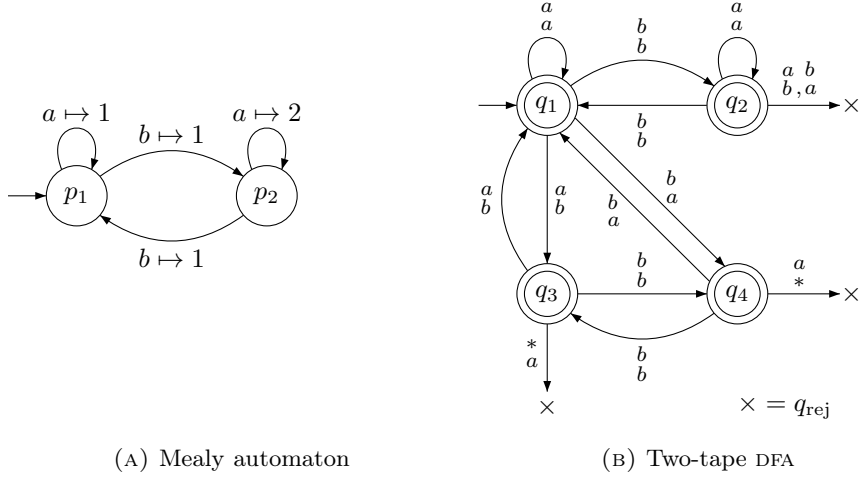

  \centering
  \begin{subfigure}[b]{.48\linewidth}

    \begin{gpicture}(42,57)(0,0)
      \gasset{Nw=8,Nh=8,Nmr=4, rdist=1, loopdiam=5}
      \node[Nmarks=i, Nframe=y](q1)(10,31){$p_1$}
      \node[Nmarks=n, Nframe=y](q2)(35,31){$p_2$}
      \drawedge[ELpos=50, ELside=l, curvedepth=6](q1,q2){$b \mapsto 1$}
      \drawedge[ELpos=50, ELside=l, curvedepth=6](q2,q1){$b \mapsto 1$}
      \drawloop[ELside=l,loopCW=y, loopdiam=5, loopangle=90](q1){$a \mapsto 1$}
      \drawloop[ELside=l,loopCW=y, loopdiam=5, loopangle=90](q2){$a \mapsto 2$}
    \end{gpicture}
     \caption{Mealy automaton\label{fig:running-example-t}}
  \end{subfigure}~
  \begin{subfigure}[b]{.48\linewidth}
    \begin{gpicture}(52,57)(0,0)
      \gasset{Nw=8,Nh=8,Nmr=4, rdist=1, loopdiam=5}     
      \node[Nmarks=ir, Nframe=y](q1)(10,43){$q_1$}
      \node[Nmarks=r, Nframe=y](q2)(35,43){$q_2$}
      \node[Nmarks=r, Nframe=y](q3)(10,18){$q_3$}
      \node[Nmarks=r, Nframe=y](q4)(35,18){$q_4$}
      \node[Nmarks=n, Nframe=n, Nw=3, Nh=4](rej3)(10,3){$\times$}
      \node[Nmarks=n, Nframe=n, Nw=3](rej2)(50,43){$\times$}
      \node[Nmarks=n, Nframe=n, Nw=3](rej4)(50,18){$\times$}
      \put(35,3){\makebox(0,0)[l]{$\times = \reject$}}
      \drawedge[ELpos=50, ELside=l, curvedepth=6](q1,q2){$\fbi{b}{b}$}
      \drawedge[ELpos=50, ELside=l, curvedepth=0](q2,q1){$\fbi{b}{b}$}
      \drawedge[ELpos=50, ELside=l, ELdist=1, curvedepth=0](q1,q3){$\fbi{a}{b}$}
      \drawedge[ELpos=50, ELside=l, ELdist=1, curvedepth=6](q3,q1){$\fbi{a}{b}$}
      \drawedge[ELpos=50, ELside=l, curvedepth=0](q3,q4){$\fbi{b}{b}$}
      \drawedge[ELpos=50, ELside=l, curvedepth=6](q4,q3){$\fbi{b}{b}$}
      \drawedge[ELpos=55, ELside=l, ELdist=0.5, sxo=2, eyo=2, curvedepth=0](q1,q4){$\cfbi{b}{a}$}
      \drawedge[ELpos=55, ELside=l, ELdist=0, sxo=-2, eyo=-2, curvedepth=0](q4,q1){$\cfbi{b}{a}$}
      \drawedge[ELpos=55, ELside=l, curvedepth=0](q2,rej2){$\fbi{a}{b}\fbi{}{,\:\!}\fbi{b}{a}$}
      \drawedge[ELpos=52, ELside=r, ELdist=1, curvedepth=0](q3,rej3){$\cfbi{*}{a}$}
      \drawedge[ELpos=55, ELside=l, curvedepth=0](q4,rej4){$\cfbi{a}{*}$}
      \drawloop[ELside=l,loopCW=y, loopdiam=5, loopangle=90](q1){$\cfbi{a}{a}$}
      \drawloop[ELside=l,loopCW=y, loopdiam=5, loopangle=90](q2){$\cfbi{a}{a}$}
    \end{gpicture}
    \caption{Two-tape \dfa\label{fig:running-example-a}}
  \end{subfigure}
  \vspace{2pt}
  \caption{A Mealy automaton and a two-tape \dfa over alphabet $\Gamma = \{a,b\}$ 
describing the same information partition (the symbol $*$ stands for $\{a,b\}$) \label{fig:geom}}
\end{figure}

To illustrate, \figurename~\ref{fig:running-example-t} shows a Mealy automaton defining
an observation function. 
The input alphabet is the set $\Gamma = \{a,b\}$ of moves,
and the output alphabet is the set $\{1,2\}$ of observations. 
For example, the histories $abb$ and $bba$ map to the same observation sequence, namely $111$,
thus they belong to the same information set; the information partition on
histories of length $2$ is $\{aa,ab,bb\},\{ba\}$.\medskip

This formalism captures the standard approach
for describing information in finite-state systems (see, e.g.,
Reif~\cite{Reif84}, Lin and Wonham~\cite{LinWon1988},
Kupferman and Vardi~\cite{KupfermanVar2000}, 
van der Meyden and Wilke~\cite{vanderMeydenWil2005}).

\subsubsection{Indistinguishability}

As a second alternative, we represent information partitions as 
equivalence relations between histories, such that the
equivalence classes correspond to information sets.
Intuitively, a player cannot distinguish between equivalent histories.

We say that an equivalence relation is an \emph{indistinguishability} relation if the
represented information partition satisfies the conditions of synchronous perfect recall.
The following characterisation simply rephrases the relevant
conditions for partitions in terms of equivalence relations.

\begin{lemma}\label{lem:ind-char} An equivalence relation $R \subseteq \Gamma^* \times \Gamma^*$
  is an indistinguishability relation if, and only if, it satisfies the following properties:
\begin{enumerate}[(1)]
\item{\label{it:ind-sync}} For every pair $(\tau, \tau') \in R$, the histories $\tau, \tau'$ are of the same length.
\item\label{it:ind-prefcl} For every pair of histories $\tau, \tau' \in R$ of length $\ell$,
  every pair $(\rho, \rho')$ of histories of length~$t \le \ell$
  that occur as prefixes of $\tau, \tau'$, respectively,
  is also related by~$(\rho, \rho') \in R$.   
\end{enumerate}   
\end{lemma}

As a finite-state representation,
we will consider indistinguishability relations recognised by two-tape automata.
To illustrate, \figurename~\ref{fig:running-example-a} shows a two-tape automaton
that defines the same information partition as the Mealy automaton of \figurename~\ref{fig:running-example-t}.
Here and throughout the paper, the state $\reject$ represents a rejecting sink state.
For example, the pair of words $\tau_1,\tau_2$ where $\tau_1 = abb$
and $\tau_2 = bba$ is accepted by the automaton (the state $q_1$ is accepting), 
meaning that the two words are indistinguishable. \medskip

Given a two-tape automaton 
$\calA = \tuple{Q, \Gamma \times \Gamma, q_\init, \delta, F}$,
the recognised relation $L(\calA)$ is, by definition, synchronous and
hence satisfies condition~(\ref{it:ind-sync}) of Lemma~\ref{lem:ind-char}. 
To decide whether $\calA$ indeed represents an indistinguishability relation, we can 
use standard automata-theoretic techniques to verify that
$L(\calA)$ is an equivalence relation, and that it 
satisfies the perfect-recall condition~(\ref{it:ind-prefcl}) of Lemma~\ref{lem:ind-char}.

\begin{lemma}
  The question whether a given two-tape automaton 
  recognises an indistinguishability relation 
  with perfect recall is decidable in polynomial (actually, cubic) time. 
\end{lemma}

The idea of using finite-state automata to describe information constraints of players in infinite games
has been advanced in a series of work by Maubert and different
coauthors~\cite{Maubert14,MaubertPin2014,BozzelliMauPin2015,DimaMauPin2018},
with the aim of
extending the classical framework of temporal logic and automata for perfect-information
games to more expressive structures.
In the general setup, the formalism features binary relations between histories
that can be asynchronous and may not satisfy perfect recall.
The setting of synchronous perfect recall is adressed as a particular case described by a one-state
automaton that compares observation sequences rather than move histories.  
This allows to capture indistinguishability relations that
actually correspond to regular observation functions in our setup. 

Another approach of  relating game histories via automata has been
proposed recently by Fournier and Lhote~\cite{FournierLho2019}.
The authors extend our framework to arbitrary
synchronous relations, which are not necessarily prefix closed\,---\,and thus do not
satisfy perfect recall.

\subsubsection{Equivalent representations}

In general, any partition of a set~$X$ can be represented either as 
an equivalence relation on $X$ ---\,equating the elements of each part\,--- 
or as a (complete) invariant function, that is a function~$f\colon X \to Z$ such that
$f( x ) = f( y )$ if, and only if, $x, y$ belong to the same part.
Thus equivalence relations and invariant functions represent different
faces of the same mathematical object. The correspondence is witnessed by the following canonical maps.

For every function $f\colon X \to Z$, the \emph{kernel} relation 
\begin{align*}
  \ker f := \{(x,y) \in X \times X~\mid~ f(x) = f(y) \}
\end{align*} is an equivalence. 
Given an equivalence relation~$\mathop{\sim} \subseteq X \times X$, 
the \emph{quotient map} $[\,\cdot\,]_\sim\colon X \to 2^X$, which 
sends each element $x \in X$ to its equivalence class $[x]_\sim := \{ y \in X ~\mid~y \sim x\}$,
is a complete invariant function for~$\sim$. 
Notice that the kernel of the quotient map is just $\mathop{\sim}$.

For the case of information partitions with synchronous perfect recall,
the above correspondence relates indistinguishability relations and
observation-history functions.

\begin{lemma}\label{lem:correspondence-gen}
  If $\beta\colon \Gamma^* \to \Sigma$ is an observation function, 
  then $\ker \hat{\beta}$ is an indistinguishability relation
  that describes the same information partition.   
  Conversely, if $\mathop{\sim}$ is an indistinguishability relation, 
  then the quotient map is an observation function
  that describes the same information partition.
\end{lemma}

Accordingly, every information partition given by an indistinguishability relation
can be alternatively represented by an observation function, and vice versa.
However, if we restrict to finite-state representations,
the correspondence might not be preserved.
In particular, 
as the quotient map of any indistinguishability relation on $\Gamma^*$
has infinite range (histories of different length are always distinguishable),
it is not definable by a Mealy automaton, which has finite output alphabet.  

\section{Observation is Weaker than Distinction}


Firstly, we shall see that, for every regular observation function, 
the corresponding indistinguishability relation is also regular.

\begin{proposition}\label{prop:transducer-to-two-tape}
  For every observation function $\beta$ given by a Mealy automaton of size $m$, 
  we can construct a two-tape \dfa of size $O(m^2)$ that defines the
  corresponding indistinguishability relation $\ker \hat{\beta}$.
\end{proposition}

\begin{proof} 
  To construct such a two-tape automaton, we
  run the given Mealy automaton on the two input tapes simultaneously,
  and send it into a rejecting sink state whenever the observation output 
  on the first tape differs from the output on the second tape.
  Accordingly, the automaton accepts a pair
  $(\tau, \tau') \in (\Gamma \times \Gamma)^* $ of histories,
  if and only if, their observation histories are equal
  $\hat{\beta}( \tau ) = \hat{\beta}( \tau' )$.    
\end{proof}

The statement of Proposition~\ref{prop:transducer-to-two-tape} is illustrated
in \figurename~\ref{fig:geom} where the structure of the two-tape \dfa of 
\figurename~\ref{fig:running-example-a} is obtained as a parallel product of
two copies of the Mealy automaton in \figurename~\ref{fig:running-example-t},
where $q_1 = \tuple{p_1,p_1}$, $q_2 = \tuple{p_2,p_2}$, $q_3 = \tuple{p_1,p_2}$, 
and $q_4 = \tuple{p_2,p_1}$.\medskip
  
For the converse direction, however, the model of imperfect information described by 
regular indistinguishability relations is strictly more expressive than the one 
based on regular observation functions.

\begin{figure}[!tb]
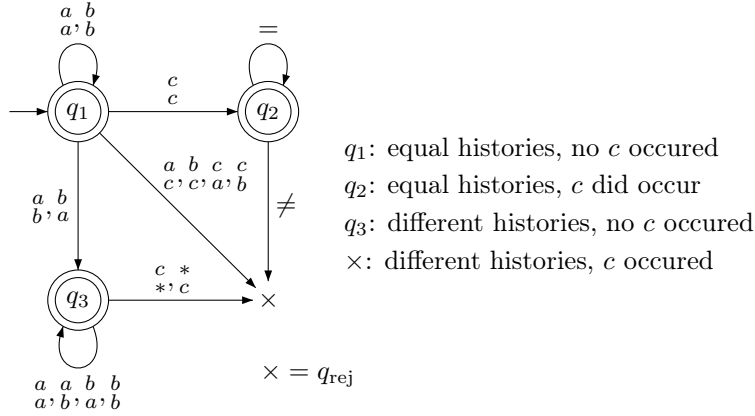

  \begin{center}

    \begin{gpicture}(100,62)(0,0)
      \gasset{Nw=8,Nh=8,Nmr=4, rdist=1, loopdiam=5}      

      \node[Nmarks=ir, Nframe=y](q1)(10,45){$q_1$}
      \node[Nmarks=r, Nframe=y](q2)(35,45){$q_2$}
      \node[Nmarks=r, Nframe=y](q3)(10,20){$q_3$}
      \node[Nmarks=n, Nframe=n, Nw=4, Nh=5](q4)(35,20){$\times$}
      \put(34,10){\makebox(0,0)[l]{$\times = \reject$}}
      
      \put(45,40){\makebox(0,0)[l]{$q_1\text{: equal histories, no } c \text{ occured}$}}
      \put(45,35){\makebox(0,0)[l]{$q_2\text{: equal histories, } c \text{ did occur}$}}
      \put(45,30){\makebox(0,0)[l]{$q_3\text{: different histories, no } c \text{ occured}$}}
      \put(45,25){\makebox(0,0)[l]{$\times\text{: different histories, } c \text{ occured}$}}

      \drawedge[ELpos=50, ELside=l, curvedepth=0](q1,q2){$\fbi{c}{c}$}

      \drawedge[ELpos=50, ELside=r, ELdist=1, curvedepth=0](q1,q3){$\fbi{a}{b},\fbi{b}{a}$}

      \drawedge[ELpos=50, ELside=l, ELdist=.5, curvedepth=0](q1,q4){$\cfbi{a}{c},\fbi{b}{c},\fbi{c}{a},\fbi{c}{b}$}
      \drawedge[ELpos=50, ELside=l, ELdist=1, curvedepth=0](q2,q4){$\neq$}
      \drawedge[ELpos=50, ELside=l, curvedepth=0](q3,q4){$\cfbi{c}{*},\cfbi{*}{c}$}

      \drawloop[ELside=l,loopCW=y, loopdiam=5, loopangle=90](q1){$\cfbi{a}{a},\fbi{b}{b}$}
      \drawloop[ELside=l,loopCW=y, ELdist=1, loopdiam=5, loopangle=90](q2){$=$}   
      \drawloop[ELside=l,loopCW=y, loopdiam=5, loopangle=270](q3){$\cfbi{a}{a},\fbi{a}{b},\fbi{b}{a},\fbi{b}{b}$}
    \end{gpicture}
  \caption{A two-tape \dfa 
defining an indistinguishability relation that does not correspond to any regular observation function (the symbol $=$ stands for $\{\tbi{a}{a}, \tbi{b}{b}, \tbi{c}{c}\}$,   
the symbol $\neq$ stands for $\{ \tbi{x}{y} \in \Gamma \times \Gamma \mid x \neq y\}$,
and the symbol $*$ stands for $\{a,b,c\}$) \label{fig:infinite}}
  \end{center}
\end{figure}

\begin{lemma}\label{lem:example}
There exists a regular indistinguishability relation 
that does not correspond to any regular observation function. 
\end{lemma}

\begin{proof}
As an example, consider a move alphabet with three letters 
$\Gamma := \{ a, b, c \}$, and let  
$\mathop{\sim} \in \Gamma^* \times \Gamma^*$ relate two histories $\tau, \tau'$ 
whenever they are equal or none of them contains the letter $c$. 
This is an indistinguishability relation, and it is recognised
by the two-tape automaton of \figurename~\ref{fig:infinite}.

We argue that the induced information tree has unbounded branching.
All histories of the same length $n$ that do not contain $c$ are indistinguishable,
hence $U_n = \{a,b\}^n$ is an information set. 
However, for every history $w \in U_n$ the history $w c$ forms a singleton information set.
Therefore $U_n$ has at least $2^n$ successors in the information tree, for every $n$.

However, for any observation function,
the degree of the induced information tree 
is bounded by the size of the observation alphabet.
Hence, the information partition described by~$\sim$
cannot be represented by an observation function of finite range and so,  
a fortiori, not by any regular observation function.
\end{proof}

\section{Which Distinctions Correspond to Observations}

We have just seen, as a necessary condition
for an indistinguishability relation to
be representable by a regular observation function,
that the information tree needs to be of bounded branching. 
In the following, we show that this condition is actually sufficient.

\begin{theorem}\label{thm:char}
  Let $\Gamma$ be a finite set of moves. 
  A regular indistinguishability relation~$\mathop{\sim}$ 
  admits a representation as a regular observation function 
  if, and only if, the information tree $\Gamma^*/_{\mathop{\sim}}$
  is of bounded branching.
\end{theorem}

\begin{proof}
  The \emph{only-if}-direction is immediate. If 
  for an indistinguishability relation~$\mathop{\sim}$,
  there exists an
  observation function $\beta\colon \Gamma^+ \to \Sigma$ with finite range  
  (not necessarily regular)  such that 
  $\mathop{\sim} = \ker \hat{\beta}$, then the maximal degree of the
  information tree $\Gamma^*/_{\mathop{\sim}}$ is at most $\abs{\Sigma}$. 
  Indeed, the observation-history function $\hat{\beta}$
  is a strong homomorphism from the
  move tree $\Gamma^*$ to the tree of observation histories $\hat{\beta}( \Gamma^*) \subseteq \Sigma^*$:
  it maps every pair $(\tau, \tau c)$ of successive move histories to the pair of
  successive observation histories $( \hat{\beta}( \tau ), \hat{\beta}( \tau ) \beta ( \tau c ))$,
  and conversely, for every pair of successive observation histories, there exists a
  pair of successive move histories that map to it.   
  By the Homomorphism Theorem (in the general formulation of Mal'cev~\cite{Malcev1973}), 
  it follows that the information tree
  $\Gamma^*/_{\mathop{\sim}} = \Gamma^*/_{\ker \hat{\beta}}$ is isomorphic to the image $\hat{\beta}( \Gamma^*)$,
  which, as a subtree $\Sigma^*$, has degree at most $\abs{\Sigma}$. 
  
  To verify the \emph{if}-direction, consider 
  an indistinguishability relation $\mathop{\sim}$ over $\Gamma^*$, 
  given by a $\dfa$~$\calR$, such that the information tree $\Gamma^*/_{\mathop{\sim}}$ 
  has branching degree at most~$n \in \mathbb{N}$. 
  
  Let us fix an arbitrary linear ordering $\preceq$ of $\Gamma$. 
  First, we pick as a representative
  for each information set,
  its least element with respect to the lexicographical order $<_\mathrm{lex}$
  induced by $\preceq$.
  Then, we order the 
  information sets in $\Gamma^*/_{\mathop{\sim}}$ according to the 
  lexicographical order of their representatives. 
  Next, we define the
  \emph{rank} of any nonempty history $\tau c \in \Gamma^*$ to be the
  index of its information set $[\tau c]_{\mathop{\sim}}$ 
  in this order, restricted to successors of $[\tau]_{\mathop{\sim}}$ ---\,this index is bounded by $n$.
  Let us consider the observation function $\beta$ that associates to every history its rank.
  We claim that (1)~it describes the same information partition as $\mathop{\sim}$ and
  (2)~it is a regular function. 

  To prove the first claim, we show that
  whenever two histories are indistinguishable $\tau \sim \tau'$, 
  they yield the same observation sequence $\hat{\beta}( \tau ) = \hat{\beta} (\tau')$. 
  The rank of a history is determined by its information set.
  Since $\tau \sim \tau'$, every pair
  $(\rho, \rho')$ of prefix histories of the same length are also indistinguishable,
  and therefore yield the same rank $\beta( \rho ) = \beta (\rho')$. 
  By definition of $\hat{\beta}$,
  it follows that $\hat{\beta}( \tau ) = \hat{\beta}( \tau')$. 
  Conversely, to verify that $\hat{\beta}( \tau ) = \hat{\beta} (\tau')$ implies
  $\tau \sim \tau'$, we proceed by induction on the length of histories. 
  The basis concerns only the empty history and thus holds trivially.
  For the induction step, suppose $\hat{\beta}( \tau c ) = \hat{\beta} (\tau' c')$.
  By definition of~$\hat{\beta}$,
  we have in particular $\hat{\beta}( \tau ) = \hat{\beta} (\tau')$,
  which by induction hypothesis implies $\tau \sim \tau'$.
  Hence, the information sets of 
  the continuations $\tau c$ and $\tau' c'$ are
  successors of the same information set $[\tau]_{\mathop{\sim}} = [\tau']_{\mathop{\sim}}$
  in the information tree $\Gamma^*/\mathop{\sim}$. 
  As we assumed that 
  the histories $\tau c$ and $\tau' c'$ have 
  the same rank, it follows that they indeed belong to the same information set,
  that is $\tau c \sim \tau' c'$.

To verify the second claim on the regularity of the observation function~$\beta$, 
we first notice that the following languages are regular:
\begin{itemize}
\item the (synchronous) lexicographical order
  $\{ (\tau, \tau') \in (\Gamma \times \Gamma)^* ~\mid~ \tau \le_\mathrm{lex} \tau' \}$,
\item the set of representatives $\{ \tau \in \Gamma^* ~\mid~
  \tau \le_\mathrm{lex} \tau' \text{ for all } \tau' \sim \tau \}$, and
\item the representation relation
  $\{ (\tau, \tau') \in ~\mathop{\sim} ~\mid~ \tau' \text{ is a representative} \}$.
\end{itemize}
Given automata recognising these languages,
we can then construct, for each $k \le n$, an automaton~$\calA_{k}$ that recognises
the set of histories of rank at least $k$: 
together with the representative of the input history,
guess the $k-1$ representatives that are below in the lexicographical order. 
Finally, we take 
the synchronous product of the automata $\calA_1 \ldots \calA_k$
and equip it with an output function as follows:
for every transition in the product automaton  
all components of the target state, up to some index $k$, are accepting ---\,we define the output of the transition
to be just this index~$k$.
This yields a Mealy automaton that outputs the rank of the input history, as desired.  
\end{proof}

For further use, we estimate the size of the Mealy automaton defining the rank function
as outlined in the proof.
Suppose that an indistinguishability relation~$\mathop{\sim} \subseteq (\Gamma \times \Gamma)^*$
given by a two-tape \dfa~$\calR$ of size $m$
gives rise to an information tree $\Gamma^*/_{L( \calR )}$ of degree~$n$.
The lexicographical order is recognisable by a two-tape \dfa of size $O( \abs{\Gamma}^2 )$, bounded by $O( m )$;
to recognise the set of representatives we take the synchronous product of this automaton with $\calR$, 
and apply a projection and a complementation, obtaining a \dfa of size bounded by $2^{O( {m^2})} )$;
for the representation relation, we take a synchronous product of this automaton with $\calR$ and obtain
a two-tape \dfa of size still bounded by $2^{O(m^2)}$. 
For every index $k \le n$, the automaton $\calA_k$ can be constructed via projection
from a  synchronous product of $n$ such automata, hence its size is bounded by $2^{2^{O(n m^2)}}$.
The Mealy automaton for defining the rank runs all
these $n$ automata synchronously, so it is of the same order of magnitude $2^{2^{O(n m^2)}}$.

To decide whether the information tree represented by a regular indistinguishability relation
has bounded degree, we use a result from the
theory of word-automatic structures~\cite{KhoussainovNer95,BlumensathGra00}.
For the purpose of our presentation, we define an automatic presentation
of a tree $T = (V, E)$ 
as a triple $(\calA_V, \calA_=, \calA_E)$ of automata
with input alphabet $\Gamma$, 
together with a surjective naming map $h \colon L \to V$
defined on a set of words $L \subseteq \Gamma^*$ such that
\begin{itemize}
\item $L( \calA_V ) = L$,
\item $L( \calA_= ) = \ker h$, and
\item $L( \calA_E ) = \{ (u, v) \in L \times L \mid (h( u ), h( v )) \in E \}$. 
\end{itemize}
In this case, $h$ is an isomorphism between $T = (V, E)$ and the
quotient $(L, L( \calA_E )) /_{L( \calA_=)}$.
The size of such an automatic presentation is the added size
of the three component automata.
A tree is automatic if it has an automatic presentation.

For an information partition given by a 
indistinguishability relation~$\mathop\sim$  
 defined by a two-tape-\dfa~$\calR$ on a move alphabet~$\Gamma$,
 the information tree $\Gamma^*/_{\mathop{\sim}}$ admits an automatic presentation
 with the naming map that sends every history $\tau$ to its
 information set $[\tau]_{\mathop{\sim}}$, and
 \begin{itemize}
 \item
   as domain automaton~$A_V$, the one-state automaton accepting all of $\Gamma^*$ (of size $\Gamma$);
 \item
   as the equality automaton $\calA_=$, the two-tape~\dfa $\calR$, and
 \item
   for the edge relation, a two-tape \dfa $\calA_E$ that recognises the relation
   \begin{align*}
     \{(\tau, \tau' c) \in \Gamma^* \times \Gamma^*~\mid~(\tau, \tau') \in L( \calR )\}.
   \end{align*}
 \end{itemize}
 The latter automaton is obtained from $\calR$ by adding transitions
 from each accepting state,  
 with any move symbol on the first tape
 and the padding symbol on the second tape, to
 a unique fresh accepting state from which all outgoing transitions
 lead to the rejecting sink~$\reject$.
 Overall, the size of the presentation will thus be bounded by $O( \abs{\calR})$. 

 Now, we can apply the following result of Kuske and Lohrey.

\begin{proposition}(\cite[Propositions 2.14--2.15]{KuskeLoh11})\label{prop:KuskeLoh}
 The question whether an automatic structure has 
bounded degree is decidable in exponential time.
If the degree of an automatic structure is bounded, 
then it is bounded by $2^{2^{{m}^{O(1)}}}$ in the size $m$ of the presentation.
\end{proposition}

This allows to conclude that the criterion of Theorem \ref{thm:char} characterising
regular indistinguishability relations that are representable by regular observation functions 
is effectively decidable.
By following the construction for the rank function outlined in the proof
of the theorem, 
we obtain a fourfold exponential upper bound for the size of
a Mealy automaton defining an observation function.

\begin{theorem}\label{thm:decision}
\begin{enumerate}[(i)]
\item The question whether an  
indistinguishability relation given as a two-tape $\dfa$ 
admits a representation as a 
regular observation function 
is decidable in exponential time (with respect to the size of the \dfa).
\item Whenever this is the case, we can construct a
  Mealy automaton of fourfold-exponential size and with at most doubly exponentially many
  output symbols 
  that defines a corresponding observation function.
\end{enumerate} 
\end{theorem}

\section{Improving the Construction of Observation Automata}

Theorem~\ref{thm:decision} establishes only a crude
upper bound on the size of a Mealy automaton
corresponding to a given indistinguishability \dfa.
In this section, we present a more detailed analysis
that allows to improve the construction by one exponential.

Firstly, let us point out that an exponential blowup is generally unavoidable,
for the size of the automaton 
and for its observation alphabet.

\begin{figure}[!tb]
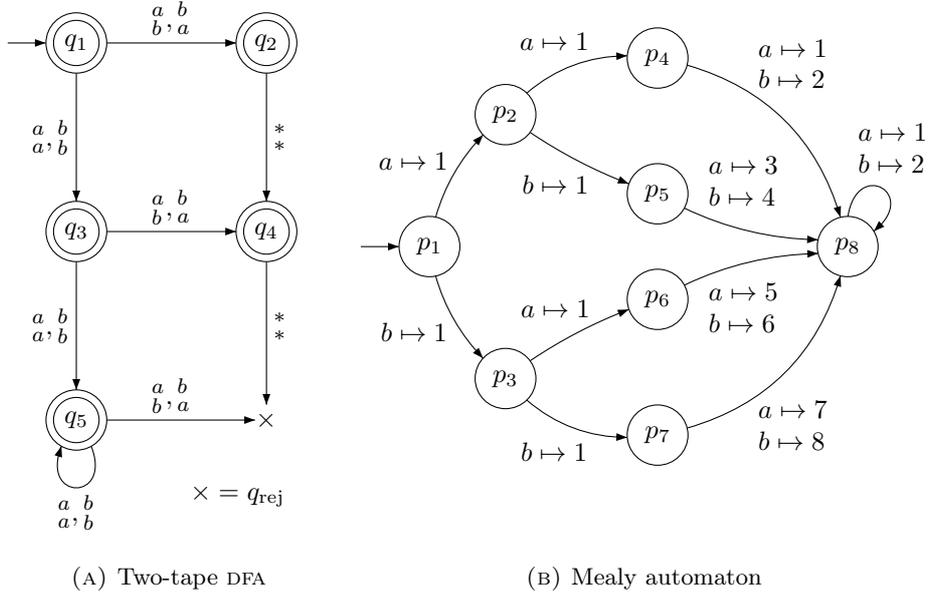

  \centering
  \begin{subfigure}[b]{0.35\textwidth}

\begin{gpicture}(42,75)(0,0)
\gasset{Nw=8,Nh=8,Nmr=4, rdist=1, loopdiam=5}    

\node[Nmarks=ir, Nframe=y](q1)(10,67){$q_1$}
\node[Nmarks=r, Nframe=y](q2)(35,67){$q_2$}
\node[Nmarks=r, Nframe=y](q3)(10,42){$q_3$}
\node[Nmarks=r, Nframe=y](q4)(35,42){$q_4$}
\node[Nmarks=r, Nframe=y](q5)(10,17){$q_5$}
\node[Nmarks=n, Nframe=n, Nw=3, Nh=4](rej)(35,17){$\times$}
\put(25,7){\makebox(0,0)[l]{$\times = \reject$}}

\drawedge[ELpos=50, ELside=l, curvedepth=0](q1,q2){$\fbi{a}{b},\fbi{b}{a}$}
\drawedge[ELpos=50, ELside=r, ELdist=1, curvedepth=0](q1,q3){$\fbi{a}{a},\fbi{b}{b}$}
\drawedge[ELpos=50, ELside=l, ELdist=1, curvedepth=0](q2,q4){$\fbi{*}{*}$}
\drawedge[ELpos=50, ELside=l, curvedepth=0](q3,q4){$\fbi{a}{b},\fbi{b}{a}$}
\drawedge[ELpos=50, ELside=r, ELdist=1, curvedepth=0](q3,q5){$\fbi{a}{a},\fbi{b}{b}$}
\drawedge[ELpos=50, ELside=l, ELdist=1, curvedepth=0](q4,rej){$\fbi{*}{*}$}
\drawedge[ELpos=50, ELside=l, curvedepth=0](q5,rej){$\fbi{a}{b},\fbi{b}{a}$}

\drawloop[ELside=l, ELdist=1, loopCW=y, loopdiam=5, loopangle=270](q5){$\fbi{a}{a},\fbi{b}{b}$}

\end{gpicture}
    \caption{Two-tape \dfa}
    \label{fig:exponential-a}
  \end{subfigure}
 \begin{subfigure}[b]{0.62\textwidth}

    \begin{gpicture}(77,75)(-1,0)
      \gasset{Nw=8,Nh=8,Nmr=4, rdist=1, loopdiam=5}

      \node[Nmarks=i, Nframe=y](q1)(10,40){$p_1$}

      \node[Nmarks=n, Nframe=y](q2)(20,57.5){$p_2$}
      \node[Nmarks=n, Nframe=y](q3)(20,22.5){$p_3$}

      \node[Nmarks=n, Nframe=y](q4)(40,65){$p_4$}
      \node[Nmarks=n, Nframe=y](q5)(40,47){$p_5$}
      \node[Nmarks=n, Nframe=y](q6)(40,33){$p_6$}
      \node[Nmarks=n, Nframe=y](q7)(40,15){$p_7$}

      \node[Nmarks=n, Nframe=y](q8)(65,40){$p_8$}

      \drawedge[ELpos=45, ELside=l, curvedepth=2](q1,q2){$a \mapsto 1$}
      \drawedge[ELpos=45, ELside=r, curvedepth=-2](q1,q3){$b \mapsto 1$}

      \drawedge[ELpos=45, ELside=l, curvedepth=3](q2,q4){$a \mapsto 1$}
      \drawedge[ELpos=45, ELside=r, curvedepth=-1](q2,q5){$b \mapsto 1$}
      \drawedge[ELpos=45, ELside=l, curvedepth=1](q3,q6){$a \mapsto 1$}
      \drawedge[ELpos=45, ELside=r, curvedepth=-3](q3,q7){$b \mapsto 1$}

      \drawedge[ELpos=40, ELdist=-1, ELside=l, curvedepth=6](q4,q8){$\begin{array}{l}a \mapsto 1 \\b \mapsto 2\end{array} $}

      \drawedge[ELpos=40, ELdist=0, ELside=l, curvedepth=-2, eyo=1](q5,q8){$\begin{array}{l}a \mapsto 3 \\b \mapsto 4\end{array} $}

      \drawedge[ELpos=40, ELdist=0, ELside=r, curvedepth=2, eyo=-1](q6,q8){$\begin{array}{l}a \mapsto 5 \\b \mapsto 6\end{array} $}

      \drawedge[ELpos=40, ELdist=-1, ELside=r, curvedepth=-6](q7,q8){$\begin{array}{l}a \mapsto 7 \\b \mapsto 8\end{array} $}

      \drawloop[ELpos=47, ELside=l, ELdist=-1, loopCW=y, loopdiam=5, loopangle=60](q8){$\begin{array}{l}a \mapsto 1 \\b \mapsto 2\end{array} $}

    \end{gpicture}
    \caption{Mealy automaton}   
    \label{fig:exponential-t}
  \end{subfigure}
  \caption{A synchronous two-tape automaton 
with $2k$ states (here $k=3$) for
which an equivalent observation \dfal requires exponential number of states ($2^k$)\label{fig:exponential}}
\end{figure}

\begin{example}

Figure~\ref{fig:exponential-a} shows a two-tape \dfa
that compares histories over a move alphabet $\{a,b\}$ with
an embargo period of length $k$. 
Every pair of histories of length less than~$k$ is accepted, whereas 
history pairs of length $k$ and onwards are rejected if, and only if,
they are different (the picture illustrates the case for $k=3$).
A~Mealy automaton that describes this indistinguishability relation 
needs to produce, for every different prefix of length $k$, a different observation
symbol.
To do so, it has to store the first $k$ symbols, which 
requires~$2^k$ states and~$2^k$ observation symbols
(see \figurename~\ref{fig:exponential-t}).
\end{example}

\begin{longversion}

  We will first identify some structural properties of indistinguishability
  relations and their \dfa, and then present the concrete construction.

\end{longversion}

\subsection{Structural properties of regular indistinguishability
relations}

For the following, let us fix a move alphabet~$\Gamma$ and
a two-tape \dfa 
$\calR = \tuple{Q,\Gamma \times \Gamma,q_\init,\delta,F}$
defining an indistinguishability
relation~$L( \calR ) = \mathop\sim$.
For convenience, we will usually write $\delta(q,\tbi{\tau\phantom{'}}{\tau'})$ for
$\delta(q,(\tau,\tau'))$. 

We assume that the automaton~$\calR$ is minimal, in the usual sense that all states are reachable
from the initial state, and the languages accepted from two different states
are different.
Note that, due to the property that whenever two
histories are distinguishable, 
their continuations are also distinguishable, 
minimality of $\calR$ also implies that
all its states are accepting, except for the single sink state
$\reject$, that is, $F = Q \setminus \{\reject\}$.

First, we classify the states according to the behaviour of the automaton 
when reading the same input words on both tapes. On the one hand,
we consider the states reachable from the initial state
on such inputs, which we call \emph{reflexive} states:
\begin{align*}
\Reflexive = \{q \in Q \mid \exists \tau \in \Gamma^*: \delta(q_\init,\tbi{\tau}{\tau}) = q\}.
\end{align*}
On the other hand, we consider the states from which it is 
possible to reach the rejecting sink
by reading the same input word on both tapes, which we call \emph{ambiguous} states,
\begin{align*}
 \Ambiguous = \{q \in Q \mid \exists \tau \in \Gamma^*: \delta(q,\tbi{\tau}{\tau}) = \reject\}.
\end{align*}
For instance, in the running example of \figurename~\ref{fig:geom}, the reflexive
states are $\Reflexive = \{q_1,q_2\}$ and
the ambiguous states are $\Ambiguous = \{q_3,q_4,\reject\}$. 

Since indistinguishability relations are reflexive,
all the reflexive states are accepting
and by reading any pair of identical words from a reflexive state, we always reach an accepting state.
Therefore, a reflexive state cannot be ambiguous. Perhaps less obviously,  
the converse also holds: a non-reflexive state must be ambiguous.

\begin{lemma}[Partition Lemma]\label{lem:partition}
$Q \setminus \Reflexive = \Ambiguous$.
\end{lemma}


\begin{longversion}

  \renewcommand\windowpagestuff{
    \begin{gpicture}(35,35)(0,0)
      \gasset{Nw=8,Nh=8,Nmr=4, rdist=1, loopdiam=5}   
      \node[Nmarks=n, Nframe=n, Nw=6, Nh=6](qI)(7,30){$q_\init$}
      \node[Nmarks=n, Nframe=n, Nw=6, Nh=6](q)(32,30){$q$}
      \node[Nmarks=n, Nframe=n, Nw=6, Nh=6](qv)(32,10){$q_{\tau}$}
      
      \drawedge[ELpos=50, ELside=l, curvedepth=0](qI,q){$\fbi{\tau\phantom{'}}{\tau'}$}
      \drawedge[ELpos=55, ELside=r, curvedepth=-4](qI,qv){$\fbi{\tau}{\tau}$}
    \end{gpicture}
  }

\begin{proof}
The inclusion $\Ambiguous \subseteq Q \setminus \Reflexive$
(or, equivalently, that $\Ambiguous$ and $\Reflexive$ are disjoint)
follows from the definitions and the fact that $\sim$ is a reflexive relation,
and thus $\delta(q_\init,\tbi{\tau}{\tau}) \neq \reject$ for all histories $\tau$.
\mbox{}\\

\opencutright
\begin{cutout}{1}{.7\textwidth}{0pt}{8}
  To show that $Q \setminus \Reflexive \subseteq \Ambiguous$,
  let us consider an arbitrary state $q \in Q \setminus \Reflexive$.
  By minimality of~$\calR$, the state $q$ is reachable from $q_\init$:
  there exist histories $\tau, \tau'$ such that $\delta(q_\init,\tbi{\tau\phantom{'}}{\tau'}) = q$.
  Let $q_{\tau} = \delta(q_\init,\tbi{\tau}{\tau})$ 
  be the state reached after reading $\tbi{\tau}{\tau}$ (see figure).
  Thus, $q_{\tau} \in \Reflexive$ and in particular $q_{\tau} \neq q$.
  Again by minimality of~$\calR$, the languages accepted from $q$ and $q_{\tau}$ are different.
  Hence, there
  exist histories $\pi,\pi'$ such that 
  $\tbi{\pi\phantom{'}}{\pi'}$ is accepted from $q$ and rejected from~$q_{\tau}$,
  or the other way round.
  In the former case, we have that $\tau \pi \sim \tau' \pi'$ and $\tau \pi \not\sim \tau \pi'$, which
  by transitivity of $\sim$, implies $\tau \pi' \not\sim \tau' \pi'$.
  This means that from state $q$ reading $\tbi{\pi'}{\pi'}$ leads to $\reject$,
  showing that $q \in \Ambiguous$, which we wanted to prove.
  In the latter case, the argument is analogous.
\end{cutout}
\end{proof}
\end{longversion}

\begin{shortversion}
\begin{proof}
The inclusion $\Ambiguous \subseteq Q \setminus \Reflexive$
(or, equivalently, that $\Ambiguous$ and $\Reflexive$ are disjoint)
follows from the definitions and the fact that $\sim$ is a reflexive relation,
and thus $\delta(q_\init,\tbi{\tau}{\tau}) \neq \reject$ for all histories $\tau$.
\mbox{}\\

  To show that $Q \setminus \Reflexive \subseteq \Ambiguous$,
  let us consider an arbitrary state $q \in Q \setminus \Reflexive$.
  By minimality of~$\calR$, the state $q$ is reachable from $q_\init$:
  there exist histories $\tau, \tau'$ such that $\delta(q_\init,\tbi{\tau\phantom{'}}{\tau'}) = q$.
  Let $q_{\tau} = \delta(q_\init,\tbi{\tau}{\tau})$ 
  be the state reached after reading $\tbi{\tau}{\tau}$ (see figure).
  Thus, $q_{\tau} \in \Reflexive$ and in particular $q_{\tau} \neq q$.
  Again by minimality of~$\calR$, the languages accepted from $q$ and $q_{\tau}$ are different.
  Hence, there
  exist histories $\pi,\pi'$ such that 
  $\tbi{\pi\phantom{'}}{\pi'}$ is accepted from $q$ and rejected from~$q_{\tau}$,
  or the other way round.
  In the former case, we have that $\tau \pi \sim \tau' \pi'$ and $\tau \pi \not\sim \tau \pi'$, which
  by transitivity of $\sim$, implies $\tau \pi' \not\sim \tau' \pi'$.
  This means that from state $q$ reading $\tbi{\pi'}{\pi'}$ leads to~$\reject$,
  showing that $q \in \Ambiguous$, which we wanted to prove.
  In the latter case, the argument is analogous.
\end{proof}
\end{shortversion}

We say that a pair of histories accepted by $\calR$ is \emph{ambiguous},
if, upon reading them, the automaton $\calR$ reaches an ambiguous state other than $\reject$.
Histories $\tau,\tau'$ that form an ambiguous pair are
thus indistinguishable, so they must map to the same observation. 
However, there exists a suffix $\pi$ such that 
the extensions $\tau\cdot\pi$ and $\tau'\cdot\pi$ become distinguishable. 
Therefore, any observation automaton
for~$\calR$ has to reach two different states after reading $\tau$
and $\tau'$ since otherwise, the extensions by the suffix $\pi$ would produce
the same observation sequence, making $\tau\cdot\pi$ and $\tau'\cdot\pi$ wrongly 
indistinguishable.
The argument generalises immediately to collections of more than two histories.
We call a set of histories
that are pairwise ambiguous an \emph{ambiguous clique}.

We shall see later, in the proof of Lemma~\ref{lem:clique-size},
that if the size of ambiguous cliques is unbounded,
then the information tree $\Gamma^*/_{L( \calR )}$ has unbounded branching, and therefore 
there exists no Mealy automaton corresponding to $\calR$.
%
%
Now, we show conversely that whenever the size of the ambiguous cliques is bounded, 
we can construct such a Mealy automaton.

We say that two histories $\tau,\tau' \in \Gamma^*$ of the same length
are \emph{interchangeable}, denoted by 
$\tau \approx \tau'$, if
$\delta(q_\init,\tbi{\tau\phantom{'}}{\pi}) = \delta(q_\init,\tbi{\tau'}{\pi\phantom{'}})$,
 for all $\pi \in \Gamma^*$.
 Note that $\approx$ is an equivalence relation
 and that $\tau \approx \tau'$ implies $\delta(q_\init,\tbi{\tau\phantom{'}}{\tau'}) \in \Reflexive$.
The converse also holds.

\begin{samepage} 
\begin{lemma}\label{lem:reflexive-strong}
  For all histories $\tau,\tau' \in \Gamma^*$, 
  we have $\delta(q_\init,\tbi{\tau\phantom{'}}{\tau'}) \in \Reflexive$ if, and only if, $\tau \approx \tau'$.
\end{lemma}

\begin{proof}
  One direction,
  that $\tau \approx \tau'$ implies $\delta(q_\init,\tbi{\tau\phantom{'}}{\tau'}) \in \Reflexive$),
  follows immediately from the definitions (take $\pi = \tau'$ in the definition of interchangeable histories).

For the reverse direction, let us suppose that $\delta(q_\init,\tbi{\tau\phantom{'}}{\tau'}) \in \Reflexive$.
We will show that, for all histories $\tau''$, the states 
$q_1 = \delta(q_\init,\tbi{\tau\phantom{''}}{\tau''})$ and
$q_2 = \delta(q_\init,\tbi{\tau'\phantom{'}}{\tau''})$ accept the same language.
Towards this, let $\pi_1$, $\pi_2$ be an arbitrary pair of histories such that
$\tbi{\pi_1}{\pi_2}$ is accepted from $q_1$. 
Then, 
\begin{itemize}
\item $\tau \pi_1 \sim \tau' \pi_1$, because $\delta(q_\init,\tbi{\tau\phantom{'}}{\tau'}) \in \Reflexive$,
and from a reflexive state reading $\tbi{\pi_1}{\pi_1}$ does not lead to $\reject$ (by Lemma~\ref{lem:partition}).
\item $\tau \pi_1 \sim \tau'' \pi_2$, because $\delta(q_\init,\tbi{\tau\phantom{''}}{\tau''}) = q_1$ and 
$\tbi{\pi_1}{\pi_2}$ is accepted from $q_1$.
\end{itemize}
By transitivity of $\sim$, it follows that $\tau' \pi_1 \sim \tau'' \pi_2$,
hence $\tbi{\pi_1}{\pi_2}$ is accepted from $q_2 = \delta(q_\init,\tbi{\tau'\phantom{'}}{\tau''})$.
Accordingly, the language accepted from $q_1$ is included in the language accepted from $q_2$;
the converse inclusion holds by a symmetric argument.
Since the states $q_1$ and $q_2$ accept the same languages, and
because the automaton~$\calR$ is minimal, it follows that $q_1 = q_2$,
which means that~$\tau$ and $\tau'$ are interchangeable. 
\end{proof}
\end{samepage}

According to Lemma~\ref{lem:reflexive-strong} and because $\reject \not\in \Reflexive$,
all pairs of interchangeable histories are also indistinguishable.
 In other words, the interchangeability relation~$\mathop{\approx}$
 refines the indistinguishability relation~$\mathop{\sim}$, and 
 thus $[\tau]_{\approx} \subseteq [\tau]_{\sim}$ for all histories $\tau \in \Gamma^*$.
In the running example (\figurename~\ref{fig:geom}), 
the sets $\{aa,ab,bb \}$ and $\{ba\}$ are $\sim$-equivalence classes, and
the sets $\{aa,bb\}$, $\{ab\}$, and $\{ba\}$ are $\approx$-equivalence classes.

Let us lift the lexicographical order $\leq_\lex$ to sets of histories of the same length by
comparing the smallest word of each set: we write $S \leq S'$ if $\min S \leq_\lex \min S'$.
This allows us to rank the $\approx$-equivalence classes contained
in a $\sim$-equivalence class, in increasing order.
In the running example, if we consider the $\sim$-equivalence class $\{aa,ab,bb \}$,
$\{aa,bb\}$ gets rank $1$, and $\{ab\}$ gets rank $2$ because 
$\{aa,bb\} \leq \{ab\}$. On the other hand, the $\sim$-equivalence class $\{ba\}$, as a singleton,
gets rank $1$.

Now, we denote by $\index(\tau)$ the rank of the $\approx$-equivalence class containing $\tau$.
For example, $\index(bb) = 1$ and $\index(ab) = 2$.
Further, we denote by $\matrix(\tau)$ the square matrix of dimension 
$n = \max_{\tau' \in [\tau]_{\sim}} \index(\tau')$
where we associate to each coordinate $i = 1, \ldots, n$ the $i$-th $\approx$-equivalence class $C_i$
contained in $[\tau]_{\sim}$. The $(i,j)$-entry of $\matrix(\tau)$
is the state $q_{ij} = \delta(q_\init,\tbi{\tau_i}{\tau_j})$ where $\tau_i \in C_i$
and $\tau_j \in C_j$. Thanks to interchangeability, 
the state $q_{ij}$ is well defined
being independent of the choice of $\tau_i$ and $\tau_j$.

\begin{longversion}

\begin{example}
In the running example, we have $a \sim b$ thus $\matrix(a) = \matrix(b)$:\medskip

$
\matrix(a) = \matrix(b) = \bordermatrix{
                & {\scriptstyle \{a\}} & {\scriptstyle \{b\}}  \cr
{\scriptstyle \{a\}} & q_1             & \!\!\!\!q_3         \cr
{\scriptstyle \{b\}} & q_4             & \!\!\!\!q_2     
         }. \\
$\medskip
\mbox{}\\
Moreover $[aa]_{\approx} = \{aa,bb\}$, and $[ab]_{\approx} = \{ab\}$, and $[ba]_{\approx} = \{ba\}$,
and thus:\medskip

$
\matrix(aa) = \matrix(ab) = \matrix(bb) = \begin{pmatrix}q_1 & q_3         \\
       q_4             & q_2     \end{pmatrix}         
\text{ and } \matrix(ba) = \begin{pmatrix}q_2 \end{pmatrix}
$.\medskip
\mbox{}\\
Note that the non-diagonal entries $q_3$ and $q_4$ are ambiguous states. This
is true in general.
\end{example}
\end{longversion}

It is easy to see that diagonal entries in such matrices are reflexive states (Lemma~\ref{lem:reflexive-strong}). 
We can show conversely that non-diagonal entries are ambiguous states.

\begin{lemma}\label{lem:non-diagonal-is-ambiguous}
For all histories $\tau$, the non-diagonal entries in $\matrix(\tau)$
are ambiguous states.
\end{lemma}

\begin{proof}
Non-diagonal entries in $\matrix(\tau)$ correspond to pair of histories that are
not $\approx$-equivalent, therefore
those entries are not reflexive states (Lemma~\ref{lem:reflexive-strong}),
hence they must be ambiguous states (Lemma~\ref{lem:partition}).
\end{proof}

\begin{longversion}
Next, we show how to construct, given $\index(\tau)$ and $\matrix(\tau)$, for some
history $\tau$, and a move $\a \in \Gamma$, the index and matrix
$\index(\tau \a)$ and $\matrix(\tau \a)$. The construction is independent of
$\tau$.

First, given a $n\times n$ matrix $M$ with entries in $Q$, 
we define $\transform(M)$ to be the $n\cdot\abs{\Gamma}\times n\cdot\abs{\Gamma}$ 
matrix obtained by substituting each entry $q_{ij}$ in $M$ with the
$\abs{\Gamma}\times \abs{\Gamma}$ matrix where every $(\a,\b)$-entry
is $\delta(q_\init,\tbi{\a}{\b})$, as illustrated in the following example. 

\begin{example}
In the running example, the $\abs{\Gamma}\times \abs{\Gamma}$ matrix
associated with state $q_1$ is:

$$ q_1 \mapsto \left(\begin{array}{cc}\delta(q_1,\tbi{a}{a}) & \delta(q_1,\tbi{a}{b}) \\[+1pt]
\delta(q_1,\tbi{b}{a}) & \delta(q_1,\tbi{b}{b}) \end{array}\right) = 
\left(\begin{array}{cc} q_1 & q_3 \\[+1pt] q_4 & q_2 \end{array}\right). $$

The matrices associated with the other states are (where we denote the $\reject$ state
by $\times$):
$$ q_2 \mapsto \left(\begin{array}{cc} q_2 & \times \\ \times & q_1 \end{array}\right)  \quad
q_3 \mapsto \left(\begin{array}{cc} \times & q_1 \\ \times & q_4 \end{array}\right)   \quad
q_4 \mapsto \left(\begin{array}{cc} \times & \times \\ q_1 & q_3 \end{array}\right). 
$$

Hence for $M = \left(\begin{array}{cc} q_1 & q_3 \\ q_4 & q_2 \end{array}\right)$,
we have $\transform(M) = \left(\begin{array}{cccc} q_1 & q_3 & \times & q_1 \\ q_4 & q_2 & \times & q_4 \\
\times & \times & q_2 & \times \\ q_1 & q_3 & \times & q_1
\end{array}\right).$
\end{example}

Second, for every $n\times n$ matrix $M$ with entries in $Q$, every $i \in \{1,\ldots, n\}$, and
every move $a \in \Gamma$, we define $\successor_{\a}(M,i) = (N,j)$, by the following construction:
\begin{itemize}
\item[$(i)$]   Initialise $N = \transform(M)$; consider the $(a,a)$
  entry of the $\abs{\Gamma} \times \abs{\Gamma}$ matrix
  substituting the $(i,i)$-entry of $M$ in 
  $N$, and initialise $j$ to be its position on the diagonal of $N$;
\item[$(ii)$] for every $1 \leq k \leq n\cdot\abs{\Gamma}$, if the $(k,j)$-entry
  of $N$ is the $\reject$ state,
  then remove the $k$-th row and $k$-th column (note that
  the $j$-th row and $j$-th column are never removed) and
  update the index $j$ accordingly;
\item[$(iii)$] if two columns of $N$ are identical, then remove the column and the corresponding row
at the larger position. If the removed column is at the position $j$,
assign the (smaller) position of the remaining duplicate column to  $j$. Repeat this step
until no two columns are identical. Return the final value of the $N$ and $j$.
\end{itemize}

\begin{example}
Consider $M =  \left(\begin{array}{cc} q_1 & q_3 \\ q_4 & q_2 \end{array}\right)$
and $i = 2$, which are the matrix and index of the history $\tau = b$ in the running example. 
In figures, the index $i$ is depicted as a vertical arrow pointing to the $i$th column of the matrix.
We obtain $\successor_{a}(M,i)$ (the matrix and index of $\tau'= ba$) as follows:
\begin{align*}
\bordermatrix{ &    & \!\!\!\!\downarrow  \cr
          & q_1     & \!\!\!\!q_3         \cr
          & q_4     & \!\!\!\!q_2     
         }
\xrightarrow{(i)}
\bordermatrix{ \!\!\!   &        & \!\!\!\!       & \!\!\!\!\downarrow & \!\!\!\!   \cr
     \!\!\!             & q_1    & \!\!\!\!q_3    & \!\!\!\!\times & \!\!\!\! q_1 \cr 
     \!\!\!             & q_4    & \!\!\!\!q_2    & \!\!\!\!\times & \!\!\!\! q_4 \cr
     \!\!\!             & \times & \!\!\!\!\times & \!\!\!\!q_2    & \!\!\!\! \times \cr
     \!\!\!             & q_1    & \!\!\!\!q_3    & \!\!\!\!\times & \!\!\!\! q_1 
         }
\xrightarrow{(ii)}
\bordermatrix{ \!\!\!   & \downarrow  \cr
     \!\!\!             & q_2 \cr 
         }
\xrightarrow{(iii)}
\bordermatrix{ \!\!\!   & \downarrow  \cr
     \!\!\!             & q_2 \cr 
         }
\end{align*}
and we obtain $\successor_{b}(M,i)$ (the matrix and index of $\tau'=bb$) as follows:
\begin{align*}
\bordermatrix{ &    & \!\!\!\!\downarrow  \cr
          & q_1     & \!\!\!\!q_3         \cr
          & q_4     & \!\!\!\!q_2     
         }
\xrightarrow{(i)}
\bordermatrix{ \!\!\!   &        & \!\!\!\!       & \!\!\!\!       & \!\!\!\! \downarrow  \cr
     \!\!\!             & q_1    & \!\!\!\!q_3    & \!\!\!\!\times & \!\!\!\! q_1 \cr 
     \!\!\!             & q_4    & \!\!\!\!q_2    & \!\!\!\!\times & \!\!\!\! q_4 \cr
     \!\!\!             & \times & \!\!\!\!\times & \!\!\!\!q_2    & \!\!\!\! \times \cr
     \!\!\!             & q_1    & \!\!\!\!q_3    & \!\!\!\!\times & \!\!\!\! q_1 
         }
\xrightarrow{(ii)}
\bordermatrix{ \!\!\!   &        & \!\!\!\!       &  \!\!\!\! \downarrow  \cr
     \!\!\!             & q_1    & \!\!\!\!q_3    &  \!\!\!\! q_1 \cr 
     \!\!\!             & q_4    & \!\!\!\!q_2    &  \!\!\!\! q_4 \cr
     \!\!\!             & q_1    & \!\!\!\!q_3    &  \!\!\!\! q_1 
         }
\xrightarrow{(iii)}
\bordermatrix{ \!\!\!   & \downarrow & \!\!\!\!  \cr
     \!\!\!             & q_1    & \!\!\!\!q_3   \cr 
     \!\!\!             & q_4    & \!\!\!\!q_2   \cr
         }.
\qedhere
\end{align*}
\end{example}

With the successor function along moves defined in this way, we obtain an
homomorphic image of $\Gamma^*$ on matrix-index pairs.   

\begin{lemma}\label{lem:successor}
For all histories $\tau \in \Gamma^*$ and moves $c \in \Gamma$,
if $(M, i) = (\matrix(\tau), \index( \tau ))$,
then $\successor_{c}(M,i) = (\matrix(\tau c), \index(\tau c))$.
\end{lemma}

\begin{proof}
The result follows from the following remarks:
\begin{itemize}
\item In step $(i)$, since $M = \matrix(\tau)$ we can associate to each row/column
of $M$ an $\approx$-equivalence class (contained in $[\tau]_{\sim}$),
say $C_1, C_2, \dots, C_n$. For $\b \in \Gamma$, and $C$ an $\approx$-equivalence class,
let $C\b = [w\b]_{\approx}$ for $w \in C$ (which is independent of the choice
of $w$ and thus well-defined - it is easy to prove that $w \approx w'$ implies 
$w \b \approx w' \b$).
We can associate to the rows/columns of $\transform(M)$
the $\approx$-equivalence classes $C_j \b$ (for each $1\leq j \leq n$ and $\b \in \Gamma$)
in lexicographic order. The stored index is the index of the $\approx$-equivalence class
of $\tau \a$.
\item In step $(ii)$, we remove the rows/columns associated with an 
$\approx$-equivalence class that is not contained in $[\tau\a]_{\sim}$
The stored index (pointing to the $\approx$-equivalence class containing $\tau\a$) 
is updated accordingly.
\item In step $(iii)$, we merge identical rows/columns which correspond 
to identical $\approx$-equivalence classes. Keeping the leftmost class
ensures the lexicographic order between $\approx$-equivalence classes is preserved. 
At the end, each $\approx$-equivalence class contained in $[\tau\a]_{\sim}$
is indeed associated to some row/column, and the resulting matrix is $\matrix(\tau \a)$
with the correct index $\index(\tau \a)$.\qedhere
\end{itemize}
\end{proof}
\end{longversion}

\subsection{Constructing the observation automaton}

\begin{longversion}
For the remainder of the paper, let us  
fix an alphabet $\Gamma$ and a
two-tape \dfa~$\calR = \tuple{Q,\Gamma \times \Gamma, \delta,q_\init,F}$
such that
the branching degree of the information tree $\Gamma^*/_{L( \calR )}$ is bounded.
Let $m$ be the size of $\calR$.
\end{longversion}

\begin{shortversion}
For the remainder of the paper, let us assume
that the branching degree of the information tree $\Gamma^*/_{L( \calR )}$ is bounded. 
\end{shortversion}


We define a Mealy automaton $\calF = \tuple{P,\Gamma, \Obs, p_\init, \delta, \lambda}$ over 
the input alphabet $\Gamma$ and an output alphabet $\Obs$ in two phases: 
first, we define the semi-automaton $\calF_0 = (P, \Gamma, p_\init,\delta) $
and then we construct the output alphabet $\Obs$ and
the output function $\lambda$. To define the semi-automaton $\calF_0$, we set:   
\begin{itemize}
\item $P := \{ \tuple{M,i} \mid M = \matrix(\tau) \text{ and } i = \index(\tau) \text{ for some history } \tau\}$,

\item $p_\init := \tuple{q_\init,1}$,

\item for every state $\tuple{M,i} \in P$ and every move $c \in \Gamma$, let $\delta(\tuple{M,i}, c) = 
\successor_{c}(M,i)$.
\end{itemize}

\begin{shortversion}
The construction of the Mealy automaton for the two-tape \dfa of \figurename~\ref{fig:running-example-a}
is shown in \figurename~\ref{fig:running-example-tc-transitions}. 
The variables $x,y,z,r,s,t,u,v$ represent the observation values
of the output function. We determine the value of the variables 
by considering pairs of histories in the automaton, and in the \dfal. 
For example, for $\tau = a$ and $\tau' = b$, we have $\tau \sim \tau'$ (according
to the \dfa), and therefore we derive the constraint $x = y$ in the \dfal.
We can show that the constraints are satisfiable and that every satisfying assignment describes 
an output function $\lambda\colon P \times \Gamma \to \Obs$ 
such that $\tuple{P,\Gamma, \Obs,p_\init, \delta,\lambda}$ is an 
observation automaton equivalent to the \dfa (see \figurename~\ref{fig:running-example-tc-observations} for the running example).
\end{shortversion}

According to Lemma~\ref{lem:successor}, the state space $P$ is the closure of $\{p_\init\}$ under
the $c$-successor operation, for all $c \in \Gamma$. It remains to show that $P$ is finite.
The key is to bound the dimension of the largest matrix in $P$, which is the size of the largest ambiguous clique.



\begin{lemma}\label{lem:clique-size}
If the branching degree of the information tree $\Gamma^*/_{L( \calR )}$ is bounded, 
then the largest ambiguous clique contains at most a doubly-exponential number of histories
(with respect to the size of $\calR$).
\end{lemma}

\begin{proof}
First we show by contradiction that the size of the ambiguous cliques is bounded.
Since the number of ambiguous states in $\calR$ is finite, if there exists an arbitrarily large ambiguous clique, 
then by Ramsey's theorem~\cite{Ramsey30}, there exists an arbitrarily large set $\{\tau_1, \tau_2, \ldots,\tau_k\}$ of histories and a state 
$q \in \Ambiguous \setminus \{\reject\}$ such that $\delta(q_\init,\tbi{\tau_i}{\tau_j}) = q$ for all $1 \leq i < j \leq k$. 
By definition of~$\Ambiguous$, there exists a nonempty history $\tau c$ 
such that $\delta(q,\tbi{\tau c}{\tau c}) = \reject$. Consider such a history $\tau c$ of minimal length. 
The histories $\tau_i \tau$ ($i= 1,\ldots,k$) are in the same $\mathop\sim$-equivalence
class, but the equivalence classes $[\tau_i \tau c]_{\mathop{\sim}}$ are pairwise distinct.
Therefore, the number of successors of $[\tau_i \tau]_{\mathop{\sim}}$ is at least $k$,
thus arbitrarily large, in contradiction with the assumption that the branching degree 
the information tree $\Gamma^*/_{L( \calR )}$ is bounded.

Note that the size of the largest ambiguous clique corresponds to the maximum number of $\approx$-equivalence classes contained in an
$\sim$-equivalence class (Lemma~\ref{lem:non-diagonal-is-ambiguous}). We show that this number is at most doubly-exponential. 
Similarly to the proof of Theorem~\ref{thm:char},
we notice that the set of \emph{$\approx$-representatives} defined by 
$\{ \tau \in \Gamma^* \mid \tau \leq_\mathrm{lex} \tau' \text{ for all } \tau' \approx \tau \}$ is regular,
and therefore the representation relation $\{ (\tau, \tau') \in \mathop{\sim} \mid \tau' \text{ is a $\approx$-representative} \}$
is also regular. Using a result of Weber~\cite[Theorem~2.1]{Weber89}, there is a bound on the number of $\approx$-representatives
that a history can have that is exponential in the size $\ell$ of the two-tape \dfa recognising the representation relation, 
namely $O(\ell)^\ell$, and $\ell$ is bounded by $2^{O(m^2)}$ by the same argument as in the proof of Theorem~\ref{thm:char}
(where $m$ is the size of $\calR$). This provides a doubly-exponential bound $2^{2^{O(m^2)}}$ on the size of the ambiguous cliques.
\end{proof}

According to Lemma~\ref{lem:clique-size}, the dimension $k$ of the largest matrix in $P$ is at most doubly exponential in $\abs{\calR}$. 
The number of matrices of a fixed dimension $d$ is at most $\abs{Q}^{d^2}$.
Overall the number of matrices that appear in $P$ is therefore bounded by $k \cdot \abs{Q}^{k^2}$,
and as the index is at most $k$, it follows that the number of states in $P$ is bounded by $k^2 \cdot \abs{Q}^{k^2}$, 
that is exponential in $k$ and triply exponential in the size of $\calR$.





\begin{figure}[!tb]
  \begin{center}
  \centering
  \begin{subfigure}[b]{.48\linewidth}
    \begin{gpicture}(65,60)(1,0)
      \gasset{Nframe=n, Nw=8,Nh=8,Nmr=0, rdist=1, loopdiam=5}

      \node[Nmarks=i](q1)(10,30){\raisebox{16pt}{\bordermatrix{ \!\!\!   & \downarrow  \cr
            \!\!\!             & q_1 \cr 
      }}}

      \node[Nmarks=n, Nw=14, Nh=17](q2)(35,45){\raisebox{21pt}{\!\!\!\bordermatrix{ & \!\!\downarrow & \cr
            & q_1     & \!\!\!\!q_3         \cr
            & q_4     & \!\!\!\!q_2     
      }}}
      \node[Nmarks=n, Nw=14, Nh=17](q3)(35,15){\raisebox{21pt}{\!\!\!\bordermatrix{ &    & \!\!\!\!\downarrow  \cr
            & q_1     & \!\!\!\!q_3         \cr
            & q_4     & \!\!\!\!q_2     
      }}}

      \node[Nmarks=n](q4)(60,15){\raisebox{16pt}{\bordermatrix{ \!\!\!   & \downarrow  \cr
            \!\!\!             & q_2 \cr 
      }}}

      \drawedge[ELpos=50, ELside=l, curvedepth=0](q1,q2){$a \mapsto x$}
      \drawedge[ELpos=50, ELside=r, curvedepth=0](q1,q3){$b \mapsto y$}

      \node[Nmarks=n, Nw=14, Nh=18](q2)(35,48){}
      \node[Nmarks=n, Nw=14, Nh=18](q3)(35,15){}

      \drawedge[ELpos=50, ELside=l, curvedepth=6, sxo=-2,exo=-2](q2,q3){$b \mapsto r$}
      \drawedge[ELpos=50, ELside=l, curvedepth=6, sxo=2,exo=2](q3,q2){$b \mapsto t$}

      \drawedge[ELpos=53, ELside=l, curvedepth=0](q3,q4){$a \mapsto s$}
      \drawedge[ELpos=50, ELside=l, curvedepth=18](q4,q1){$b \mapsto v$}

      \node[Nmarks=n, Nw=12, Nh=12](q2)(35,45){}

      \drawloop[ELside=l,loopCW=y, loopdiam=5, loopangle=45, loopdiam=4](q2){$a \mapsto z$}
      \drawloop[ELside=l,loopCW=y, loopdiam=5, loopangle=270, loopdiam=4](q4){$a \mapsto u$}
    \end{gpicture}   
     \caption{Transition structure \label{fig:running-example-tc-transitions}}
  \end{subfigure}~
  \begin{subfigure}[b]{.48\linewidth}
    \begin{gpicture}(65,60)(1,0)
      \gasset{Nframe=y, Nw=8,Nh=8,Nmr=4, rdist=1, loopdiam=5}

      \node[Nmarks=i](q1)(10,30){}
      \node[Nmarks=n](q2)(35,45){}
      \node[Nmarks=n](q3)(35,15){}
      \node[Nmarks=n](q4)(60,15){}

      \drawedge[ELpos=50, ELside=l, curvedepth=0](q1,q2){$a \mapsto 1$}
      \drawedge[ELpos=50, ELside=r, curvedepth=0](q1,q3){$b \mapsto 1$}

      \drawedge[ELpos=50, ELside=l, curvedepth=6, sxo=-1,exo=-1](q2,q3){$b \mapsto 1$}
      \drawedge[ELpos=50, ELside=l, curvedepth=6, sxo=1,exo=1](q3,q2){$b \mapsto 1$}

      \drawedge[ELpos=53, ELside=l, curvedepth=0](q3,q4){$a \mapsto 2$}
      \drawedge[ELpos=50, ELside=l, curvedepth=18](q4,q1){$b \mapsto 2$}

      \drawloop[ELside=l,loopCW=y, loopdiam=5, loopangle=90, loopdiam=4](q2){$a \mapsto 1$}
      \drawloop[ELside=l,loopCW=y, loopdiam=5, loopangle=90, loopdiam=4](q4){$a \mapsto 1$}
    \end{gpicture}
      \caption{Instantiated observations \label{fig:running-example-tc-observations}}
   \end{subfigure}
  \end{center}
  \caption{Construction of the \dfal from the two-tape \dfa of \figurename~\ref{fig:running-example-a} \label{fig:running-example-tc}}
\end{figure}

\begin{longversion}
The construction of the Mealy automaton for the two-tape \dfa of \figurename~\ref{fig:running-example-a}
is shown in \figurename~\ref{fig:running-example-tc-transitions}. 
The variables $x,y,z,r,s,t,u,v$ represent the (currently) unknown observation values
of the output function. We will build a system of constraints over these variables
by considering pairs of histories in the automaton, and in the \dfal. 
For example, for $\tau = a$ and $\tau' = b$, we have $\tau \sim \tau'$ (according
to the automaton), and therefore we derive the constraint $x = y$ in the \dfal.

We are now ready to define the output function.
Towards this, we associate to each state $p \in P$
and letter $\a \in \Gamma$, a variable $x_{p,\a}$ intended to represent the
output value $\lambda(p,\a)$.
We gather all constraints that these variables should satisfy to describe
a valid output function, and we show that the constraints are satisfiable.

For the semi-automaton $\calF_0$
defined so far, consider
the parallel product $\calF_0 \prl \calF_0$ (which is a  semi-automaton over the alphabet
$\Gamma \times \Gamma$), and the synchronised product of $\calF_0 \prl \calF_0$
with $\calR$ (thus again a semi-automaton over alphabet $\Gamma \times \Gamma$).

Our constraints are either equality or disequality
between variables. We construct a set~$\Phi$ of constraints by looking at
the synchronised product $(\calF_0 \prl \calF_0) \times R$:
for every reachable state $((p_1,p_2),q)$ with $q \neq \reject$ 
and all letters $\a, \b \in \Gamma$ (possibly $\a = \b$),
if $\delta(q,\tbi{\a}{\b}) \neq \reject$, then add
the constraint $x_{p_1,\a} = x_{p_2,\b}$ to $\Phi$, otherwise add
the constraint $x_{p_1,\a} \neq x_{p_2,\b}$ to $\Phi$.

\begin{example}
We obtain the following set of constraints
for the \dfal of \figurename~\ref{fig:running-example-tc} (we omit
trivial constraints such as $x=x$):
\begin{align*}
  \begin{array}{ll@{\quad }|@{\quad }l}
    x=y &\text{ witnessed by }  a \sim b   & s \neq t \text{ witnessed by } ba \not\sim bb \\
    t=z &\text{ witnessed by } aa \sim bb  & u \neq v \text{ witnessed by } baa \not\sim bab\\
    r=t &\text{ witnessed by } ab \sim bb  & z \neq s \text{ witnessed by } aa \not\sim ba \\
    z=r &\text{ witnessed by } aa \sim ab  & r \neq s \text{ witnessed by } ab \not\sim ba \\
  \end{array}
\end{align*}
which is equivalent to the set of constraints $\{x=y, z=r=t, t \neq s, u \neq v\}$ and
is satisfiable, e.g., with the following assignment (see \figurename~\ref{fig:running-example-tc-observations}):
\begin{align*}
  \begin{array}{l@{\qquad\qquad }l@{\qquad\qquad }l}
    x=y=1   & s=2 & u=1 \\
    z=r=t=1 &     & v=2\\
  \end{array}
\end{align*}
\end{example}

\begin{lemma}\label{lem:constraints}
\begin{itemize}
\item The set $\Phi$ of constraints is satisfiable (over any infinite domain).
\item Every satisfying assignment for $\Phi$ describes 
an output function $\lambda\colon P \times \Gamma \to \Obs$ 
such that $\tuple{P,\Gamma, \Obs,p_\init, \delta,\lambda}$ is an 
observation automaton equivalent to $\calR$.
\end{itemize}
\end{lemma}

\begin{proof}

  For the first point, it is sufficient to show that no contradiction occurs 
in~$\Phi$, namely that the following situations are impossible: $\Phi$ contains
the constraint $x_1 \neq x_k$ and a chain of equalities between variables 
$x_1 = x_2$, $x_2 = x_3, \ldots, x_{k-1} = x_k$. Towards a contradiction,
suppose that such a situation occurs ---\,with $k=3$ for simplicity of presentation,
the argument generalises straightforwardly to every finite $k$\,--- and
assume $x_{p,\a} = x_{r,\b} = x_{s,\gamma}$ and $x_{p,\a} \neq x_{s,\gamma}$
are constraints in $\Phi$. It follows that:
\begin{enumerate}[(1)]
\item there exist histories $u_1$, $u_2$ such that 
	\begin{itemize}
	\item $p = \delta(p_\init, u_1)$, 
	\item $r = \delta(p_\init, u_2)$,
	\item $u_1 \a \sim u_2 \b$;
	\end{itemize}
\item there exist histories $v_2$, $v_3$ such that 
	\begin{itemize}
	\item $r = \delta(p_\init, v_2)$, 
	\item $s = \delta(p_\init, v_3)$,
	\item $v_2 \b \sim v_3 \gamma$;
	\end{itemize}
\item there exist histories $w_1$, $w_3$ such that $w_1 \sim w_3$ and
  \begin{itemize}
	\item $p = \delta(p_\init, w_1)$, 
	\item $s = \delta(p_\init, w_3)$,
	\item $w_1 \a \not\sim w_3 \gamma$.
	\end{itemize}

  \end{enumerate}

Note that the states $p$ and $r$ differ only by their index, not by their matrix
(by Lemma~\ref{lem:successor} because $u_1 \sim u_2$, and thus $\matrix(u_1) = \matrix(u_2)$),
analogously for states $r$ and $s$. Hence, for some matrix $M$
we can write $p = \tuple{M,m_1}$, $r = \tuple{M,m_2}$, and $s = \tuple{M,m_3}$.
Then it follows from Lemma~\ref{lem:successor} and the definitions of $\matrix(\cdot)$
and $\index(\cdot)$ that (denoting by $M(i,j)$ the $(i,j)$-entry of $M$):
\begin{itemize}
\item $M(m_1,m_2) = \delta(q_\init,\tbi{u_1}{u_2})$,
\item $M(m_2,m_3) = \delta(q_\init,\tbi{v_2}{v_3})$,
\item $M(m_1,m_3) = \delta(q_\init,\tbi{w_1}{w_3})$.
\end{itemize}

Now consider, in the $\sim$-equivalence class $[u_1]_{\sim}$ of $u_1$,
the $m_3$-th $\approx$-equivalence class $C$, and a word $u_3 \in C$.
Then $\matrix(u_3) = M$ and $\index(u_3) = m_3$, thus 
$s = \tuple{M,m_3} = \delta(p_\init, u_3)$. It follows that:
\begin{itemize}
\item $M(m_2,m_3) = \delta(q_\init,\tbi{u_2}{u_3})$,
\item $M(m_1,m_3) = \delta(q_\init,\tbi{u_1}{u_3})$,
\end{itemize}
and therefore $u_2 \b \sim u_3 \gamma$ and $u_1 \a \not\sim u_3 \gamma$,
which together with $u_1 \a \sim u_2 \b$ contradicts the transitivity
of $\sim$.
Hence, we can conclude that the constraint set $\Phi$ is satisfiable.

For the second point, fix a satisfying assignment for the constraints in $\Phi$.
Take the set of values assigned to the variables as the (finite) output alphabet $\Obs$,
and define the output function by $\lambda(p,\a) = x_{p,\a}$.

We show by induction on the length of histories
that the indistinguishability relation induced by the \dfal 
is the same as the one defined by~$\calR$. The base case is trivial. 
For the induction step, let us consider an arbitrary pair $\tau, \tau'$ of histories of the same length,
under the induction hypothesis, $\hat{\lambda}( \tau)  = \hat{\lambda}( \tau' )$ if, and only if, $\tau \sim \tau'$
(according to the automaton~$\calR$). For any pair $\a, \b \in \Gamma$ of letters, if $\tau \not\sim \tau'$,
then $\tau a \not\sim \tau' b$ and $\hat{\lambda}( \tau a)  \neq \hat{\lambda}( \tau' b)$. Else, if $\tau \sim \tau'$, 
let $p = \delta(p_\init, \tau)$ and $p' = \delta(p_\init, \tau')$ be the states
reached in the semi-automaton $\calF_0$ after reading $\tau$ and $\tau'$, 
and let $q = \delta(q_\init,\tbi{\tau\phantom{'}}{\tau'})$ be the state reached in the 
automaton $\calR$ after reading the pair $(\tau,\tau')$. It follows
that the state $((p,p'),q)$ is reachable in the synchronised product $(\calF_0 \prl \calF_0) \times R$.
Here, we distinguish two cases:
\begin{itemize}
\item if $\tau \a \sim \tau' \b$, then 
the constraint $x_{p,\a} = x_{p',\b}$ is in $\Phi$,
and therefore the observation of $\a$ in state $p$ is
the same as the observation of $\b$ in state $p'$ 
($\lambda(p,\a) = \lambda(p',\b)$).
\item if $\tau \a \not\sim \tau' \b$, then 
the constraint $x_{p,\a} \neq x_{p',\b}$ is in $\Phi$,
and therefore the observation of $\a$ in state $p$ is
different from the observation of $\b$ in state $p'$ 
($\lambda(p,\a) \neq \lambda(p',\b)$).
\end{itemize}
In either case, we thus have $\hat{\lambda}(\tau a) = \hat{\lambda}(\tau' b)$ if, and only if, $\tau a \sim \tau' b$,
which concludes the proof.
\end{proof}

\end{longversion}

\begin{longversion}
 
  Lemma~\ref{lem:constraints} establishes the correctness of the constructed Mealy automaton~$\calF$.
  Since the size of $\calF$ is exponential in the size $k$
  of the largest ambiguous clique, 
  and~$k$ is at most doubly-exponential (Lemma~\ref{lem:clique-size}),
  we get the following result.
\end{longversion}

\begin{theorem}\label{theo:characterization}
  For every indistinguishability relation given by a two-tape \dfa $\calR$
  such that the information tree  $\Gamma^*/_{L( \calR )}$ is of bounded branching, we can construct 
  a Mealy automaton of triply exponential size (with respect to the size of $\calR$)
  that defines a corresponding observation function.
\end{theorem}


\begin{longversion}
\end{longversion}

\section{Conclusion}

The question of how to model information in infinite games
is fundamental to defining their strategy space.
As the decisions of each player are based on the available information,
strategies are functions from information sets to actions. Accordingly, the information
structure of a player in a game defines the support of her strategy space.

The assumption of synchronous perfect recall gives rise
to trees as information structures (Lemma~\ref{lem:info-tree}).
In the case of observation functions with a finite range
$\Sigma$, these trees are subtrees of the complete $\abs{\Sigma}$-branching
tree $\Sigma^*$\,---\,on which $\omega$-tree automata can work (see~\cite{Thomas95,LNCS2500} for surveys
on such techniques).
Concretely, every strategy based on observations can be represented as
a labelling of the tree $\Sigma^*$ with actions; the set of all strategies for a 
given game forms a regular (that is, automata-recognisable) set of trees.
Moreover, when considering winning conditions that are also regular,  
Rabin's Theorem~\cite{Rabin72} allows to conclude
that winning strategies also form a regular set.
Indeed, we can construct effectively a tree automaton that recognises the set of strategies
(for an individual player) that enforce a regular condition and,
if this set is non-empty, we can also synthesise a Mealy automaton
that defines one of these strategies.
In summary, the interpretation of strategies as observation-directed trees
allows us to search the set of all strategies systematically for winning
ones using tree-automatic methods.    

In contrast, when setting out with indistinguishability relations, we obtain more complicated
tree structures that do not offer a direct grip to classical tree-automata techniques.
As the example of Lemma~\ref{lem:example} shows, there are cases where the
information tree of a game is not regular, 
and so the set of all strategies is not recognisable by a tree automaton.
Accordingly, the automata-theoretic approach to strategy synthesis
via Rabin's Theorem cannot be applied to solve, for instance,
the basic problem of constructing a finite-state strategy for
one player to enforce a given regular winning condition.

On the other hand, modelling information with indistinguishability relations
allows for significantly more expressiveness than observation functions. 
This covers notably settings where a player can receive an
unbounded amount of information in one round. 
For instance, models with causal memory where one player may communicate his entire observation history
to another player in one round can be captured with regular indistinguishability relation,
but not with observation functions of any finite range.
Even when an information partition that can be represented by
finite-state observation functions, the representation
by an indistinguishability relation may be considerably
more succinct.  
For instance, a player that observes the move history perfectly,
but with a delay of~$d$ rounds can be described 
by a two-tape \dfa with $O( d )$ many states,
whereas any Mealy automaton would require exponentially more states
to define the corresponding observation function.

At the bottom line, as a finite-state model of information,
indistinguishability relations are strictly more expressive and
can be (at least exponentially) more succinct than observation functions.
In exchange, the observation-based model is directly accessible to automata-theoretic methods,
whereas the indistinguishability-based model is not.
Our result in Theorem~\ref{thm:decision} allows to
identify effectively the instances of indistinguishability relations
for which this gap can be bridged.
That is, we may take advantage of the 
expressiveness and succinctness of
indistinguishability relations to describe a game problem
and use the procedure to obtain, whenever possible,
a reformulation in terms of observation functions towards
solving the initial problem with automata-theoretic methods.

This initial study opens several exciting research directions.
One immediate question is whether the fundamental finite-state methods on strategy synthesis
for games with imperfect information 
can be extended from the observation-based model to the one based on indistinguishability
relations. Is it decidable, given a game for one player with a
regular winning condition against Nature, whether there exist a winning strategy ?
Can the set of all winning strategies be described by finite-state automata ?
In case this set is non-empty, does it contain a strategy
defined by a finite-state automaton ? 

Another, more technical, question concerns the automata-theoretic foundations of
games. The standard  models are laid out for representations of games and strategies
as trees of a fixed branching degree.
How can these automata models be extended to trees with unbounded branching
towards capturing strategies constrained by indistinguishability relations ?
Likewise, the automatic structures that arise as information quotients
of indistinguishability relations form a particular class of trees,
where both the successor and the descendant relation (that is, the transitive closure) are regular.
On the one hand, this particularity may allow to decide properties about games (viz. their information trees)
that are undecidable when considering general automatic trees,
notably regarding bisimulation or other forms of game equivalence.

Finally, in a more application-oriented perspective, it will be worthwhile to explore
indistinguishability relations as a model for games where players
can communicate via messages of arbitrary length. In particular this will allow to
extend the framework of infinite games on finite graphs to systems
with causal memory considered in the area of distributed computing.

\bibliographystyle{plainurl} 
\bibliography{biblio}

\end{document}